\documentclass[12pt]{amsart}

\usepackage{cite}

\usepackage{kotex}
\usepackage{lineno,hyperref}
\modulolinenumbers[5]
\usepackage{epstopdf}
\usepackage{pifont}
\usepackage{tikz}
\usepackage{latexsym}
\usepackage{graphics}
\usepackage{graphicx}
\usepackage{float}
\usepackage[dvips]{xy}
\xyoption{all}
\usepackage{ifpdf}
\usepackage{multirow}
\usepackage{amssymb, amsthm, amsmath}
\usepackage{hhline}
\usepackage{centernot}
\usepackage{verbatim, comment, color}
\usepackage{enumerate}
\usepackage{setspace}
\usepackage{lipsum}
\usepackage{subcaption}
\usepackage{mathtools}

\usepackage{bbm}
\usepackage{dsfont}
\usepackage{tikz}

\newtheorem{theorem}{Theorem} 
\newtheorem{proposition}[theorem]{Proposition}
\newtheorem{lemma}[theorem]{Lemma}
\newtheorem{remark}[theorem]{Remark}
\newtheorem{corollary}[theorem]{Corollary}
\newtheorem{definition}[theorem]{Definition}
\newtheorem{example}[theorem]{Example}

\newcommand{\ba}{\begin{align}}
\newcommand{\ea}{\end{align}}  

\newcommand{\be}{\begin{equation}}

\newcommand{\ee}{\end{equation}}
\newcommand{\bea}{\begin{eqnarray}}
\newcommand{\eea}{\end{eqnarray}}
\newcommand{\barr}{\begin{array}}
\newcommand{\earr}{\end{array}}
\newcommand{\bn}{\begin{enumerate}}
\newcommand{\en}{\end{enumerate}}
\newcommand{\bi}{\begin{itemize}}
\newcommand{\ei}{\end{itemize}}
\newcommand{\bbbm}{\begin{pmatrix}}
\newcommand{\eeem}{\end{pmatrix}}

\newcommand{\bbS}{{\bf S}}

\newcommand{\cE}{{\cal E}}

\newcommand{\cP}{{\cal P}}

\newcommand{\R}{{\mathbf R}}

\newcommand{\al}{\alpha}
\newcommand{\bt}{\beta}

\newcommand{\ga}{\gamma}

\newcommand{\De}{\Delta}
\newcommand{\ep}{\epsilon}

\newcommand{\la}{\lambda}

\newcommand{\ignore}[1]{}{}
\newcommand{\noin}{\noindent}

\newcommand{\nn}{\nonumber}

\newcommand{\p}{{\partial}}

\newcommand{\q}{\quad}

 \newcommand{\Id}{\mathop{\rm Id}}

\newcommand{\Prob}{{\mathcal P}}
\newcommand{{\QED}}{{\hfill QED} \bigskip}

\newcommand{\Rn}{{\R^n}}
\newcommand{\spt}{\mathop{\rm spt}}
\newcommand{\Tr}{\mathop{\rm Tr}}

\newcommand{\Wa}{W_\alpha}
\newcommand{\Wb}{{W_\beta}}
\newcommand{\Wab}{W_{\alpha,\beta}}

\renewcommand{\subset}{\subseteq}

\newcommand{\cal}{\mathcal}

 \DeclareMathOperator*{\argmin}{argmin}
  \DeclareMathOperator*{\argmax}{argmax}

  \definecolor{darkspringgreen}{rgb}{0.09, 0.45, 0.27} 
 \definecolor{darkgray}{rgb}{0.66, 0.66, 0.66}

\newcommand{\alo}{\underline{\alpha}_{\Delta^1}}
\newcommand{\alod}{\alo^+}
\newcommand{\aln}{\underline{\alpha}_{\Delta^n}}
\newcommand{\alnd}{\aln^+}

\newcommand{\ali}{\al_\infty^*}
\newcommand{\bti}{\bt_\infty^*}
\newcommand{\fin}{f_\infty^*}
\newcommand{\kn}{4^*}

\numberwithin{equation}{section}
\numberwithin{theorem}{section}
\tikzset{
    partial ellipse/.style args={#1:#2:#3}{
        insert path={+ (#1:#3) arc (#1:#2:#3)}
    }
}

\begin{document}
\title[Classifying minimizers of attractive-repulsive interactions]
{Classifying minimum energy states for interacting particles: regular simplices
 }
\thanks{
\em \copyright 2022 by the authors.  The authors are grateful to Dejan Slepcev, Rupert Frank and {Steven Damelin} for stimulating interactions.
 CD acknowledges partial support of his research by a Natural Sciences and Engineering Research Council of Canada Undergraduate Summer Research Assistantship. 
 TL is grateful for the support of ShanghaiTech University, and in addition, to the University of Toronto and its Fields Institute for the Mathematical
Sciences, where parts of this work were performed.  RM  acknowledges partial support of his research by the Canada Research Chairs Program and
Natural Sciences and Engineering Research Council of Canada Grant {2020-04162.}}

\date{\today}

\author{Cameron Davies, Tongseok Lim and Robert J. McCann}
\address{Cameron Davies: Department of Mathematics \newline University of Toronto, Toronto ON Canada}
\email{cameron.davies@mail.utoronto.ca}
\address{Tongseok Lim: Krannert School of Management \newline  Purdue University, West Lafayette, Indiana 47907, USA}
\email{lim336@purdue.edu}
\address{Robert J. McCann: Department of Mathematics \newline University of Toronto, Toronto ON Canada}
\email{mccann@math.toronto.edu}

\begin{abstract} 
Densities of particles on $\Rn$ which interact pairwise through an attractive-repulsive power-law potential $W_{\al,\bt}(x) = |x|^\al/\al-|x|^\bt/\bt$ have often been used to explain patterns produced by biological and physical systems.
 In the mildly repulsive regime 
 {$\al> \bt \ge 2$ with $n \ge 2$, we 
show there exists 
a decreasing homeomorphism $\al_{\Delta^n}$ from $[2,4]$ to itself such that:} 
 distributing the particles uniformly over the vertices of a regular unit diameter $n$-simplex minimizes the potential energy if and only if $\al\ge \al_{\De^n}(\bt)$. Moreover this minimum is uniquely attained up to rigid motions when $\al > \al_{\De^n}(\bt)$. We estimate $\al_{\De^n}(\bt)$ above and below, and identify its limit as the dimension grows large. These results are derived from a new northeast comparison principle in the space of exponents. At the endpoint $(\al,\bt)=(4,2)$ of this transition curve, we characterize all minimizers by showing they lie on a sphere and share all first and second moments with the spherical shell. Suitably modified versions of these statements are also established {(i) for $W_{\alpha,\beta}$ and corresponding energies in the case where $n=1$,} and {(ii)} for the  attractive-repulsive potentials $D_\al(x) = |x|^\al(\al\log |x|-1)$ that arise in the  
 limit $\bt \nearrow \al$.
\end{abstract}


\maketitle


\section{Introduction}
Particles interacting through long-range attraction and short-range repulsion given by differences of power-laws have been used to model a range of
 physical  \cite{Lennard-Jones31}  \cite{HolmPutkaradze06}  and 
 biological  \cite{Patlak53} \cite{Breder54} \cite{KellerSegel70} systems,  to predict or explain many of the patterns they display 
 \cite{AlbiBalagueCarrilloVonBrecht14}
\cite{BertozziKolokolnikovSunUminskyVonBrecht15}
\cite{KolokolnikovSunUminskyBertozzi11}
\cite{vonBrechtUminskyKolokolnikovBertozzi12},
{and to select mesh points for numerical integration \cite{Damelin08, DamelinGrabner03, DamelinLevesleyRagozinSun09, DamelinMaymeskul05}.}
For very few values of the attractive and repulsive exponents $(\al,\bt)$ are the energy minimizing configurations of particles explicitly known; see however
 \cite{BurchardChoksiHess-Childs20} \cite{CarrilloHuang17} \cite{CarrilloShu21+} \cite{ChoksiFetecauTopaloglu15}
\cite{DaviesLimMcCann21+a} \cite{FetecauHuang13} \cite{FetecauHuangKolokolnikov11} {\cite{FrankLieb19+}} \cite{LimMcCann21}.
 Here we complement these results which --- apart from \cite{LimMcCann21} --- concern $\bt < 2$, 
 by showing that for a region containing the intersection of the 
 {quadrant $(\al,\bt) \in [4,\infty) \times [2,\infty) \setminus \{(4,2)\}$
 with the halfspace} $\al>\bt$,   the minimizer consists precisely of those configurations which equidistribute their particles over the vertices of an appropriately sized simplex,
 i.e. an equilateral triangle in two dimensions and a regular tetrahedron in three.  We are able to give a detailed description the region in question,  and explain precisely how uniqueness of these minimizers fails at its corner $(\al,\bt)=(4,2)$.

Let us recall the setting and notation from our companion work \cite{DaviesLimMcCann21+a}:
 The self-interaction energy of a collection of particles with mass distribution $d\mu(x) \ge 0$ on $\R^n$ is given by 
\begin{align} \label{energy}
\mathcal{E}_W(\mu) =  \frac12 \iint_{\R^n \times \R^n} W(x-y) d\mu(x)d\mu(y),
\end{align}
assuming the particles interact with each other through a pair potential $W(x)$.   Normalizing the distribution to have unit mass ensures that $\mu$ belongs to the space $\Prob(\R^n)$ of Borel
probability measures on $\R^n$. 

Our goal is to identify global energy minimizers
of $\cE_W(\mu)$ on $\Prob(\Rn)$,  for {\em power-law} potentials $W=\Wab$ where
\begin{align} \label{potential}
\Wa(x) &:= |x|^\al/\al  \q {\rm and}
\\ \Wab(x) &:= \Wa(x) - \Wb(x) 
\q  \alpha > \beta>-n,
\label{potential2}
\end{align}
with the appropriate convention if $\alpha=0$ or $\beta=0$ \cite{BCLR13}.
In this paper we focus {exclusively} on the mildly repulsive regime $\bt \ge 2$ of \cite{CarrilloFigalliPatacchini17}, and its frontier $\bt =2$.
The latter is called the centrifugal line in \cite{LimMcCann21}, since, at least on $\R^2$, the potential $-W_2$ induces the outward force which particles rotating uniformly around their common center of mass seem to experience in a corotating reference frame; see e.g.~\cite{McCann06}. On this frontier the energy also acts as a Lyapunov function of the rescaled dynamics of the purely attractive Patlak-Keller-Segel  \cite{Patlak53} \cite{KellerSegel70} model in self-similar variables around the time of blow-up \cite{SunUminskyBertozzi12}.
On the segment $(\al,\bt) \in (2,4) \times \{2\}$,  our companion paper shows the minimizer is uniquely given (up to translations) by a spherical shell --- i.e. the uniform probability measure on a spherical hypersurface --- at least if $n \ge 2$. 

For $\al \ge 4$ and $\al > \bt \ge 2$ but $(\al , \bt) \neq (4,2)$, the present work shows that the minimizer is uniquely given (apart from rotations and translations) by the measure $ \nu=\nu_1$ which equidistributes its mass over the vertices of a regular, unit diameter, $n$-simplex, defined below, i.e. an equilateral 
triangle if $n=2$ and a regular tetrahedron if $n=3$. These results answer a question of Sun, Uminsky and Bertozzi, by showing that the linear stability of selfsimilar blow-up 
which they found 
for the aggregation dynamics 
on the boundary of these two complementary regimes 
can be improved to a nonlinear stability result.
 This improvement is explained in \cite{DaviesLimMcCann21+a}; for spherically symmetric perturbations of the spherical shell,  such an improvement was already found by Balagu\'e et al \cite{BalagueCarrilloLaurentRaoul13N},
while asymptotic stability of measures on the simplex vertices was addressed by Simione \cite{Simione14}.
On the other hand, at the threshold exponent separating these two regimes, 
we will show that although all centered convex combinations of the configurations mentioned above remain mimimizers, 
there are many additional minimizers as well: indeed for $(\al,\bt)=(4,2)$ the centered minimizers consist precisely of all measures supported on the minimizing spherical shell which share its moments up to order 2. When $n \ge 2$, this case is distinguished from $\al \ne 4$ by the fact that the attractor formed by global energy minimizers becomes infinite-dimensional.

In the mildly repulsive region $\al > \bt \ge 2$,  two of us recently showed the existence of a finite
threshold $\al_{\Delta^n}(\beta)<\infty$ above which the energy is uniquely minimized by 
$ \nu_{1}$ and its rotates and translates \cite{LimMcCann21}.  In the current manuscript,  we estimate
$\al_{\Delta^n}(\beta) \le \max\{\bt,4\}$ concretely, showing equality holds when $\bt=2\le n$ and finding the high dimensional limiting threshold explicitly in the broader range $\bt > 2$.
We also show it is impossible for $ \nu_1$  to minimize 
$ \mathcal{E}_{W_{\al,\bt}}$ for any $\al < \al_{\Delta^n}(\beta)$. Further results concerning $\al_{\Delta^n}$ are established in \S \ref{S:threshold} below and summarized 
in {Theorem \ref{T:threshold} and Remark \ref{R:threshold}}. 

 To describe our conclusions, it will be convenient to recall the following class of sets and measures which were the main object of study in \cite{LimMcCann19p} \cite{LimMcCann21}. We say that a set $K \subset \R^n$ is called a {\rm regular $ k$-simplex} if it is the convex hull of $ k+1$ points $\{x_0, x_1,...,x_{k}\}$ in $\R^n$ satisfying $|x_i-x_j|=d$ for some $d>0$ and all $0 \le i < j \le {k}$. 
 The points $\{x_0, x_1,...,x_{k}\}$ are called {\em vertices} of the simplex. In particular, it is called a {\em unit $k$-simplex} if $d=1$. 
We also define the following set of measures:
\begin{align}\label{simplex}
\cP_{\Delta^n}:=\{\nu \in \cP(\R^n) \ | \ &\nu \text{ is uniformly distributed over }  \\
&\text{the vertices of a unit $n$-simplex.} \} \nn
\end{align}
In particular $\cP_{\Delta^1} = \{\frac12(\delta_{a} + \delta_{a+1}) \ | \ a \in \R\}$.
 Let $\cP^0_{\Delta^n}= \cP_{\Delta^n} \cap \cP_0(\Rn)$ where
 $\mathcal{P}_0(\R^n)$ denotes the centered probability measures on $\R^n$ --- meaning those having finite first moments and satisfying 
\be \label{centered}
\int_{\R^n} x \, d\mu(x)=0.
\ee  
We can now present our results. Let $\Id$ denote the $n\times n$ identity matrix.
\begin{theorem}[Characterizing energy minimizers at $(\al,\bt)=(4,2)$]
\label{main1}
A measure 
$\mu \in \cP_0(\R^n)$ minimizes $\mathcal{E}_{W_{4,2}}$ in \eqref{energy} if and only if $\mu$ is concentrated on the centered sphere of radius $ \sqrt{\frac{n}{2n+2}}$ 
 and has
\be\label{tensor}
\int x \otimes x \, d\mu(x) = \bigg( \int x_ix_j d\mu(x)\bigg)_{1 \le i,j \le n}= \frac{1}{2n+2} {\rm Id}. 
\ee
\end{theorem}
Notice, if $n=1$, $\frac{\delta_{-1/2}+\delta_{1/2}}{2}\in \mathcal{P}_{\Delta^1}$ is the only minimizer in $ \cP_0(\R)$. 
For $n=3$, several inequivalent minimizers are illustrated in Figure \ref{fig:1}.

\begin{figure}[h]
\caption{Convex hulls of supports of sample minimizers of $\mathcal{E}_{W_{4,2}}$ in $\mathcal{P}_0(\R^3)$. Each of these four minimizers is inscribed in the sphere of radius $\sqrt{3/8}$ and has mass uniformly distributed over the set of extreme points of the convex hull of its support. Moreover, rotates and convex combinations of any of these minimizers are also minimizers. This implies that general minimizers of $\mathcal{E}_{W_{4,2}}$ need not have any rotational symmetries.}
\

\centering
\begin{tikzpicture}[x  = {(-0.5cm,-0.5cm)},
                    y  = {(0.9659cm,-0.25882cm)},
                    z  = {(0cm,1cm)},
                    scale = 2]   
\draw[fill=red!50!white,opacity=0.5] 
( 0.57735026919, 0,-0.20412414523) --
( -0.28867513459, -0.5,-0.20412414523) --
(-0.28867513459, 0.5,-0.20412414523)--cycle;
\draw[fill=red!50!white,opacity=0.5] (0, 0,0.61237243569) --
( -0.28867513459, -0.5,-0.20412414523) --
(-0.28867513459, 0.5, -0.20412414523)--cycle;
\draw[fill=red!50!white,opacity=0.5] ( 0, 0, 0.61237243569) --
( 0.57735026919, 0, -0.20412414523) --
( -0.28867513459, 0.5, -0.20412414523)--cycle;
\draw[fill=red!50!white,opacity=0.5] ( 0, 0, 0.61237243569) --
( 0.57735026919, 0, -0.20412414523) --
( -0.28867513459, -0.5, -0.20412414523) --cycle;
\end{tikzpicture}  
\
\begin{tikzpicture}[scale = 2]   
\draw[fill=green!50!white,opacity=0.5] 
(-0.61237243569, 0,0) --
( 0, -0.61237243569, 0) --
(0,0, -0.61237243569)--cycle;
\draw[fill=green!50!white,opacity=0.5] 
(0.61237243569, 0,0) --
( 0, -0.61237243569, 0) --
(0,0, -0.61237243569)--cycle;
\draw[fill=green!50!white,opacity=0.5] 
(-0.61237243569, 0,0) --
( 0, 0.61237243569, 0) --
(0,0, -0.61237243569)--cycle;
\draw[fill=green!50!white,opacity=0.5] 
(0.61237243569, 0,0) --
( 0, 0.61237243569, 0) --
(0,0, -0.61237243569)--cycle;
\draw[fill=green!50!white,opacity=0.5] 
(-0.61237243569, 0,0) --
( 0, 0.61237243569, 0) --
(0,0, 0.61237243569)--cycle;
\draw[fill=green!50!white,opacity=0.5] 
(0.61237243569, 0,0) --
( 0, 0.61237243569, 0) --
(0,0, 0.61237243569)--cycle;
\draw[fill=green!50!white,opacity=0.5] 
(0.61237243569, 0,0) --
( 0, -0.61237243569, 0) --
(0,0, 0.61237243569)--cycle;
\draw[fill=green!50!white,opacity=0.5] 
(-0.61237243569, 0,0) --
( 0, -0.61237243569, 0) --
(0,0, 0.61237243569)--cycle;
\end{tikzpicture}
\ \    
\begin{tikzpicture}[x  = {(-0.5cm,-0.5cm)},
                    y  = {(0.9659cm,-0.25882cm)},
                    z  = {(0cm,1cm)},
                    scale = 2]   
\draw[fill=blue!50!white,opacity=0.5] 
( -0.35355339059, -0.35355339059, -0.35355339059) --
( 0.35355339059, -0.35355339059,-0.35355339059) --
( 0.35355339059, 0.35355339059,-0.35355339059)--
( -0.35355339059, 0.35355339059,-0.35355339059)--cycle;
\draw[fill=blue!50!white,opacity=0.5] 
( -0.35355339059, -0.35355339059, -0.35355339059) --
( 0.35355339059, -0.35355339059,-0.35355339059) --
( 0.35355339059, -0.35355339059, 0.35355339059)--
( -0.35355339059, -0.35355339059, 0.35355339059)--cycle;
\draw[fill=blue!50!white,opacity=0.5] 
( -0.35355339059, -0.35355339059, -0.35355339059) --
( -0.35355339059, -0.35355339059,0.35355339059) --
( -0.35355339059, 0.35355339059,0.35355339059)--
( -0.35355339059, 0.35355339059,-0.35355339059)--cycle;
\draw[fill=blue!50!white,opacity=0.5] 
( 0.35355339059, 0.35355339059,0.35355339059) --
( -0.35355339059, 0.35355339059,0.35355339059) --
( -0.35355339059, -0.35355339059,0.35355339059)--
( 0.35355339059, -0.35355339059,0.35355339059)--cycle;
\draw[fill=blue!50!white,opacity=0.5] 
( 0.35355339059, 0.35355339059,0.35355339059) --
( -0.35355339059, 0.35355339059,0.35355339059) --
( -0.35355339059, 0.35355339059, -0.35355339059)--
( 0.35355339059, 0.35355339059, -0.35355339059)--cycle;
\draw[fill=blue!50!white,opacity=0.5] 
( 0.35355339059, 0.35355339059,0.35355339059) --
( 0.35355339059, 0.35355339059,-0.35355339059) --
( 0.35355339059, -0.35355339059,-0.35355339059)--
( 0.35355339059, -0.35355339059,0.35355339059)--cycle;
\end{tikzpicture} 
\ \ 
\begin{tikzpicture}[scale=2]
\draw[thin] (0,0) circle (0.61237243569cm);
\draw[opacity=0.4] (0,0) [partial ellipse =0:180:0.61237243569cm and .2cm];
\fill[violet!75!white, opacity=0.5](0,0) circle (0.61237243569cm);
\draw(0,0) [partial ellipse =180:360:0.61237243569cm and .2cm];
\end{tikzpicture}
\label{fig:1}
\end{figure}

Now for each $\al>\bt$, let 
\[A_{\al,\bt} = \{(\al',\bt') \in \R^2 
\ | \ \al'>\bt',\ \al' \ge \al,\ \bt' \ge \bt,\ (\al',\bt') \ne (\al,\bt)\}
\]
denote the region of parameters lying north, east, or northeast of $(\al,\bt)$. {The following theorem allows us to extend an energy comparison involving a unit simplex from a single point $(\alpha,\beta)$ in parameter space to the entire {\em northeast} region $A_{\alpha,\beta}$ which lies above and to its right. As we learned from the referees, when $n=2$ and the interaction energy \eqref{energy} is equipped with the one-parameter family of anisotropic potentials $\Tilde{W}_{\alpha}(x):=-\log(|x|)+\alpha\frac{x_1^2}{|x|^2},$ 
an analogous comparison principle was formulated independently by
Carrillo et al. \cite{CarrilloMateuMoraRondiScardiaVerdera20}, who used it to show that the known unique minimizer of $\mathcal{E}_{\Tilde{W}_{1}}$ also uniquely minimizes $\mathcal{E}_{\Tilde{W}_{\alpha}}$ for each $\alpha\ge 1$.  In effect,  the theorem which follows provides two-parameter monotonicity results for power-law potentials analogous to their one-parameter result for anisotropic potentials.
}
\begin{theorem}[Northeast comparison of simplex energies and potentials]
\label{main2} 
Let $\al > \bt > 0$. If $ \nu \in \cP_{\Delta^n}$ minimizes $ \mathcal{E}_{W_{\al,\bt}}$ on $\cP(\Rn)$, then for $(\al',\bt') \in A_{\al,\bt}$, 
\be\label{twin argmin identities}
\cP_{\De^n} = \argmin_{\cP(\Rn)} \cE_{W_{\al',\bt'}} \quad {\rm and} \quad \spt \nu = \argmin_\Rn (\nu *W_{\al',\bt'}).
 \ee
\end{theorem}

\begin{remark}[One dimension]\label{R:1D}
If $n=1$, our companion paper \cite{DaviesLimMcCann21+a} shows $\cP_{\Delta^1}$ uniquely minimizes $ \mathcal{E}_{W_{\al,2}}$ for all $\al \ge 3$. Kang, Kim, Lim and Seo \cite[Theorem 2]{KKLS21} on the other hand showed $\cP_{\Delta^1}$ is not a $d_\infty$-local minimizer, hence not a global minimizer, in the range $\bt=2<\al<3$. 
\end{remark}
Set 
\be\label{4*}
\kn:=
\left\{
\begin{array}{cl} 
3 & {\rm if}\ n=1 \\
4 & {\rm otherwise}.
\end{array}
\right.
\ee
Notice Theorems \ref{main1}, \ref{main2} and Remark \ref{R:1D}  imply the following corollary; see also Figure \ref{fig:2}.

\begin{corollary}[Simplices minimize for $\al \ge \max\{\kn ,\bt\}$]
\label{C:main3}
 For each $(\al,\bt) \in A_{\kn,2}$, $\cP_{\Delta^n}$ uniquely minimizes $ \mathcal{E}_{W_{\al,\bt}}$ on $\cP(\R^n)$.
\end{corollary}

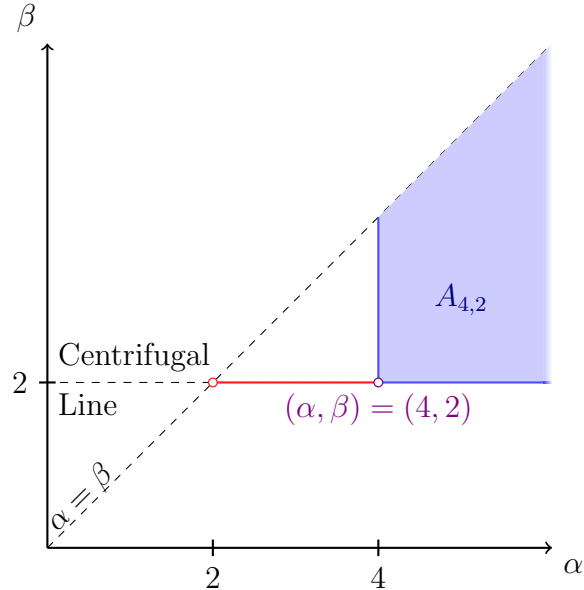
\begin{figure}[h]

\caption{Partial phase diagram of the mildly repulsive region $\al>\bt \ge 2$ for $n \ge 2$:  on the red segment linking $(2,2)$ to
$(4,2)$,  energy is uniquely minimized by a spherical shell \cite{DaviesLimMcCann21+a}.   At $(\al,\bt)=(4,2)$,  the energy is minimized by any convex combination of the {configurations described in Figure 1, but also admits other minimizers {characterized by Theorem \ref{main1}}. In the blue region, $A_{4,2}$,  Theorem \ref{main2} and the fact that the balanced unit simplices $\cP_{\Delta^n}$ minimize $\mathcal{E}_{W_{4,2}}$ combine to imply that the interaction energy} is minimized  {precisely by the elements of $\cP_{\Delta^n}$.}}

\centering

\begin{tikzpicture}[scale=1.1]
\draw[thick,->] (0,0) -- (6.1,0)node[anchor=north west] {$\alpha$};
\draw[thick,->] (0,0) -- (0,6.1)node[anchor=south east] {$\beta$};
\draw [dashed] (0,0) -- (6.1,6.1);
\draw [dashed] (6,2) -- (0,2) node[anchor=south west] {{Centrifugal}}node[anchor=north west] {Line};
\fill[blue!20!white] (4,2) -- (4,4) -- (6,6) -- (6,2) -- (4,2);
\draw[thick, red!90!white] (4,2) -- (2,2); 
\draw[thick, blue!70!white] (4,2) -- (4,4);
\draw[thick, blue!70!white] (4,2) -- (6,2);
\fill[violet] (4,2) circle (0.06cm) node[anchor=north] {$(\alpha,\beta)=(4,2)$};
\fill[white] (4,2) circle (0.045cm);
\fill[red!90!white] (2,2) circle (0.06cm);
\fill[white] (2,2) circle (0.045cm);
\node[blue!50!black] at (5,3){$A_{4,2}$};
\draw[thick] (-0.1,2) node[anchor=east] {2}
-- (0.1,2)  ;
\draw[thick] (2,-0.1) node[anchor=north] {2}
-- (2,0.1)  ;
\draw[thick] (4,-0.1) node[anchor=north] {4}
-- (4,0.1)  ;
\node[black, rotate=45] at (0.4,0.6) {$\alpha=\beta$};
\shade[left color=blue!20!white,right color=white] (6,2) -- (6,6) -- (6.1,6.1) -- (6.1,2) -- (6,2);
\shade[left color=blue!70!white,right color=white] (6,1.9905) rectangle (6.1,2.0095);
\end{tikzpicture}
\label{fig:2}
\end{figure}

{Our theorems, and in particular Theorem \ref{main2}, allow us to infer quantitative results about the structure of threshold function $\al_{\Delta^n}(\beta)$ which two of us defined implicitly in \cite[Corollary 1.4]{LimMcCann21}. For any given $\beta\ge 2,$ this threshold function describes the critical value $\alpha_{\Delta^n}(\beta)$ such that, for all $\alpha>\alpha_{\Delta^n}(\beta),$ $\mathcal{E}_{W_{\alpha,\beta}}$ is uniquely minimized by the unit simplices $\mathcal{P}_{\Delta^n}$. Prior to the present work and its companion paper \cite{DaviesLimMcCann21+a}, nothing was known of the behaviour of $\alpha_{\Delta^n}$, save for its abstract existence as a function from $[2,\infty)$ to $[2,\infty)$ and the lower bound $\alpha_{\Delta^n}(2) \ge 4$ provided by \cite[Remark 1.5]{LimMcCann21}. Our techniques now yield following much more precise statement, which implies continuity and monotonicity properties of the threshold function, and shows that if $\alpha\in (\beta,\alpha_{\Delta^n}(\beta)),$ then $\mathcal{E}_{W_{\alpha,\beta}}$ is not minimized by any unit simplex:}

{\begin{theorem}[Transition threshold]\label{T:threshold}
For each $ \bt \ge 2$ there exists $\al_{\Delta^n}(\bt) \in [\bt, \infty)$ such that  
\begin{align}
 \label{gthreshold} 
\cP_{\Delta^n}
=\argmin_{\cP(\Rn)} \mathcal{E}_{W_{\al,\bt}} 
&\mbox\ {\rm if}\ \al > \al_{\Delta^n}(\bt),
\\ \label{lthreshold}
\emptyset = \cP_{\Delta^n}
\cap \argmin_{\cP(\Rn)} \mathcal{E}_{W_{\al,\bt}} 
&\mbox\ {\rm if}\ \bt < \al < \al_{\Delta^n}(\bt). 
\end{align}
If $\al=\al_{\De^n}(\bt)$ and $\nu \in \cP_{\De^n}$, then at least 
one of the following two containments is strict:
\be \label{ethreshold}
\cP_{\Delta^n}
\subsetneq \argmin_{\cP(\Rn)} \mathcal{E}_{W_{\al,\bt}}
\quad {\rm or} \quad
\spt \nu \subsetneq \argmin_{\Rn} (W_{\al,\bt} * \nu).
\ee
Moreover, $\al_{\De^n}(2) =\kn $ from \eqref{4*}, and we have $\beta_n \in (2,\kn )$ such that $\al_{\Delta^n}(\bt) = \bt$ for $\bt \ge \bt_n$, and 
$\al_{\Delta^n}:[2,\bt_n] \longrightarrow [\bt_n, \kn]$ is continuous and strictly decreasing. 
\end{theorem}}
{The quantity $\beta_n$ defined in Theorem \ref{T:threshold} represents the smallest value of $\beta$ such that the graph of the threshold function $\alpha_{\Delta^n}(\beta)$ intersects the diagonal boundary $\alpha=\beta$ of the mildly repulsive regime in parameter space. 
Later, in Corollary \ref{C:main4}, we will see that, while $\mathcal{E}_{\alpha,\beta}$ is trivial on this boundary, we can define a non-trivial family of interaction kernels $D_\alpha$ which continuously extend the symmetrically rescaled family of energies $\frac{\al\bt}{\al-\bt}\cE_{\Wab}$ 
to the line $\alpha=\beta$. 
In the meantime, let us describe upper and lower bounds on $\alpha_{\Delta^n}(\beta),$ which will be made rigorous in subsections \ref{SS:upper} and \ref{SS:Lower}, respectively:}

\begin{remark}[Bounds on the Transition Threshold]\label{R:threshold}
{By using the same family of rescaled kernels $\frac{\al\bt}{\al-\bt}{\Wab}$, 
Definition \ref{upperbounddef} specifies a function $\alpha_{\infty}^*=\alpha_{\infty}^*(\beta)$ which for $n\ge 2$ becomes independent of dimension. Corollary \ref{threshold lower bound} shows that $\alpha_\infty^*$ bounds the threshold function $\alpha_{\Delta^n}$ from above in the sense that $\alpha_\infty^*(\beta)\ge \alpha_{\Delta^n}(\beta)$ for all $\beta\in[2,\infty).$  
Conversely, in subsection \ref{SS:Lower}, we use violations of an Euler-Lagrange equation \eqref{EL} for the interaction energy \eqref{energy} to define a pair of dimensionally-dependent lower bounds for $\alpha_{\Delta^n}$. The first, $\underline{\alpha}_{\Delta^n}^+$, defined in 
\eqref{threshold lower bound}, arises from checking whether the Euler-Lagrange equation for the unit simplex is violated anywhere in $\R^n.$ The second family of bounds, $\underline{\alpha}_{\Delta^n}$ defined in 
\eqref{threshold lower estimate}, instead arise from looking for violations of the Euler-Lagrange equation at a specific point in $\R^n$ which is chosen based on the dimension. As we show in Proposition \ref{lowerlowerboundprop}, $\underline{\alpha}_{\Delta^n}^+$ is a sharper lower bound than $\underline{\alpha}_{\Delta^n},$ and the fact that it is sensitive to Euler-Lagrange violations at each point in $\R^n$ means that, unlike $\underline{\alpha}_{\Delta^n},$ its strength does not depend on the choice of reference point. However, the theoretical appeal of a sharper bound is muted by the apparent intractability of computing a bound which requires us to check an inequality at each point of $\R^n$. On the other hand, Definition \ref{D:weaker lower bound} allows us to implicitly define $\underline{\alpha}_{\Delta^n}(\beta)$ by using a single equation (or equivalently inequality) involving $\alpha$ and $\beta,$ and with an apt choice of reference point, this bound need not be much weaker than the theoretically superior bound $\underline{\alpha}_{\Delta^n}^+.$
Moreover, Proposition \ref{P:high dimensional collapse} guarantees that even the weaker bound $\underline{\alpha}_{\Delta^n}$ is asymptotically sharp for large dimensions, in the sense that for each $\beta\ge 2,$ we have $\lim_{n\to\infty} \aln(\beta) = \ali(\beta).$} Even so,
 it would be interesting to
know the value of $\bt_n$ and of $\al_{\Delta^n}(\bt)$ in the range $\bt \in (2, \bt_n)$ more precisely. For example,  might $\al_{\De^n}\equiv\alnd$?
\end{remark}


\begin{remark}[Open global minimization problems]
An interesting open problem is to determine the structure of minimizers of $\mathcal{E}_{W_{\al,\bt}}$ for $2<\bt<\al<\al_{\Delta^n}(\bt)$. Carrillo, Figalli, and Patacchini showed the supports of such minimizers must have finite cardinality, and placed a bound on this cardinality \cite{CarrilloFigalliPatacchini17}, but little else is known about this subregime. 
If $n=1$ and $\beta=2$, {identifying the} global minimizers of $\mathcal{E}_{W_{\al,2}}$ along the segment $(\al,\bt)\in (2,3)\times \{2\}$ of the centrifugal line {was highlighted by us as another open problem in the original release of this preprint.} 
 {Shortly thereafter, the latter problem was elegantly solved by R.~Frank \cite{Frank21+}, who 
 used Fourier analysis, convexity and the Euler-Lagrange equation \eqref{EL} to show 
 the (unique centered) solution takes the form 
 $d\mu_\al(x) = C (R^2-x^2)_+^{(1-\al)/2}dx$ for certain explicit constants $C,R>0$ depending on $\al\in(2,3)$.}
\end{remark}

{
\begin{remark}[Physically realistic potentials]
The mildly repulsive regime $\alpha > \beta \ge 2$ which we address may be unphysical in several respects:  the potentials $W_{\al,\bt}(x)$ grow rapidly as $r=|x| \to \infty$
(meaning long range forces increase without bound), yet remain bounded at $r=0$, which permits a positive fraction of the particles to condense on the same point.  
These may or may not be desirable features, depending on what one is trying to model.
It is perhaps worth pointing out the global energy minimizers $\nu$ we identify for these potentials will remain $d_\infty$-local minimizers (see \cite{DaviesLimMcCann21+a} and \eqref{KRW metric})
for any other potential $W$ which agrees with $W_{\alpha,\bt}$ in a neighourhood of $|x|=0$ and of $|x|=1$ when $\al>\al_{\Delta^n}(\bt)$ (or of $x \in \spt \nu - \spt \nu$ more generally).  This includes potentials which need not be spherically symmetric, nor grow at infinity.  On the other hand,  our techniques say nothing obvious about potentials  with singularities at the origin that prevent condensation onto points,   such as $W_{\al,\bt}$
with $\bt<2$ or bond order potentials more generally.
\end{remark}
}

Finally, taking the limit $\bt \to \al$ for the rescaled potential $\overline{W}_{\al,\bt} = \frac{\al\bt}{\al-\bt}W_{\al,\bt}$ (which has minimum value $-1$), leads us to introduce the following new class of interaction kernels,
\be
{D_\al(x):= \alpha^2\frac{\p}{\p\alpha}W_{\alpha,\beta}(x) }= |x|^\al (\al\log |x| - 1), \q \al \in \R \setminus \{0\}
\ee
 which form another intriguing one-parameter family of attractive-repulsive potentials uniquely minimized at $|x|=1$. {This family continuously extends of the two-parameter family of rescaled potentials $\overline{W}_{\alpha,\beta}$ to the portion of the boundary of the mildly repulsive regime which lies on the diagonal line $\alpha=\beta.$ This interpretation is supported by the following corollary, which follows from the proof of Theorem \ref{main2} and relates the minimizers of  $\mathcal{E}_{W_{\al,\bt}}$ to those of $ \mathcal{E}_{D_{\al}}$:}

\begin{corollary}[Relation to minimizers of limiting potential {on the diagonal}]
\label{C:main4}
 If $\cP_{\Delta^n}$ minimizes $ \mathcal{E}_{W_{\al,\bt}}$ for some $\al > \bt > 0$, then $\cP_{\Delta^n}$ uniquely minimizes $\mathcal{E}_{D_\ga}$ on $\cP(\Rn)$ for all $\ga \ge \al$. 
 Conversely, if $\cP_{\De^n}$ minimizes $ \mathcal{E}_{D_\bt}$ for some $\bt > 0$, then $\cP_{\De^n}$ uniquely minimizes $\mathcal{E}_{W_{\al,\bt}}$ on $\cP(\Rn)$ for all 
 $\al>\bt$. Thus from Remark \ref{R:threshold}, $\cP_{\Delta^n}$ minimizes $\mathcal{E}_{D_\al}$  uniquely if $\al > \bt_n$, and fails to minimize $\cE_{D_\al}$ if $ 0 < \al <\bt_n$.
 \end{corollary}
 {In effect, the preceding corollary states that, if unit simplices minimize $\mathcal{E}_{W_{\alpha,\beta}}$ for some point $(\alpha,\beta)$ in the mildly repulsive regime, they also minimize $\mathcal{E}_{D_\gamma}$ for all $\gamma\ge \alpha$. By the formal relation $\mathcal{E}_{D_\gamma}=\frac{d}{d\alpha}\left.\mathcal{E}_{W_{\alpha,\beta}}\right|_{\alpha=\gamma},$ this means that, as $\gamma$ increases from $\alpha,$ the interaction energy of the unit simplex decreases more quickly (or increases more slowly) than that of any non-simplicial measure. A rigorous version of this heuristic comparison argument is crucial to the proof of Theorem \ref{main2}. On the other hand, this corollary states that, when we consider the closure $\alpha\ge\beta\ge 2$  of the mildly repulsive regime in parameter space and interpret ${\overline W}_{\alpha,\alpha}:=D_{\alpha},$ then the region on which $\mathcal{P}_{\Delta^n}$ minimizes $\overline{\mathcal{E}}_{\alpha,\beta}$ is a closed subregion. In other words, this provides us with a reasonable way of extending the threshold function to the boundary $\alpha=\beta$.}


\section{Classifying minimizers at $(\al,\bt) =(4,2)$}

Our first task is to adapt Lopes' proof \cite{L19} of energetic convexity from densities to measures in Lemma \ref{linearity},  extracting conditions for strict convexity;
 see \cite{CarrilloDelgadinoDolbeaultFrankHoffman19} and \cite{DaviesLimMcCann21+a} for the analogous extension in the interval $(\al,\bt) \in (2,4) \times \{2\}$,
whose endpoint we now analyze.

\begin{definition}[Second moment tensor] The second moment tensor for $\mu \in \cP(\R^n)$ is the $n \times n$ matrix given by
\be\label{MOI}
I(\mu)=\int x \otimes x \, d\mu(x) = \bigg( \int x_ix_j d\mu(x)\bigg)_{i, j \in \{1,\ldots, n\}}.
\ee
\end{definition}

\begin{lemma}[Moment criteria for strict convexity]
\label{linearity}
For any $\mu_0, \mu_1 \in \cP_0(\R^n)$ {having finite fourth moments}, set $a(t):=\mathcal{E}_{W_4}(\mu_t)$ where $\mu_t := (1-t)\mu_0+t\mu_1$. 
Then $a(t)$ is convex, and depends affinely on $t\in [0,1]$ if and only if $I(\mu_0) = I(\mu_1)$.
\end{lemma}

\begin{proof} 
Fix $\mu_0,\mu_1 \in \cP_0(\R^n)$ with fourth moments. Since $\mathcal{E}_{W_4}(\mu)$ is a quadratic function of $\mu$, we see $a''(t)=2\mathcal{E}_{W_4}(\mu_0-\mu_1)$.
Thus convexity and affinity of $a(t)$ on $t\in[0,1]$ depend on the sign of
\[
8\mathcal{E}_{W_4}(\mu_0-\mu_1)=\iint_{\mathbb{R}^n\times \mathbb{R}^n}|x-y|^4 d(\mu_0-\mu_1)(x)d(\mu_0-\mu_1)(y).
\]
Vanishing of the zeroth and first moments of $\eta:=\mu_0-\mu_1$
allows us to express
$\cE_{W_4}(\eta)$ 
as the following sum of squares involving the second moment tensors $I(\eta) :=I(\mu_0)-I(\mu_1)$ 
from \eqref{MOI}

\begin{align*}
8\mathcal{E}_{W_4}(\eta)
&=\iint_{\R^n\times\R^n} [4(x\cdot y)^2+ 2|x|^2|y|^2]d\eta(x)d\eta(y)
\\ &=4 \Tr(I(\eta)^2) + 2 (\Tr I(\eta))^2.
\end{align*}
Thus $\cE_{W_4}(\eta) \ge 0$ with equality if and only if $I(\mu_0)=I(\mu_1)$, as desired.
\end{proof} 

\begin{lemma}[Second moments for measures on centered spheres]
\label{sameMOI} 
Let $\bbS_r$ be the centered sphere of radius r in $\R^n$, and let $\mu \in 
\cP(\bbS_r)$. If $I(\mu) = \la {\rm Id}$ for some $\la >0$, then $I(\mu)=I(\sigma_r)$ where $\sigma_r$ is the uniform probability on $\bbS_r$.
\end{lemma}

\begin{proof}  If $I(\mu)=\la{\rm Id}$, any rotation $R\mu$ of $\mu$ has the same second moment tensor $I(R\mu)=I(\mu)$.
Now if we uniformize $\mu$ by averaging over its rotations, the resulting measure $\sigma_r$ will have the same {second} moment {tensor} as $\mu$ due to the linearity of $I$. 
\end{proof} 

It is plausible that the following lemma is known, but lacking a reference we provide a proof for the sake of clarity and completeness.

\begin{lemma}[Minimizing moments under moment constraints]
\label{moments} 
Let $0 < p < q < \infty$, $C>0$ and $\mu_0 \in \cP(\R^n)$. Then
\be
\mu_0 \in \argmin \bigg\{ \int |x|^q d\mu(x) \ \bigg| \ \mu  \in \cP(\R^n),\  \int |x|^p d\mu(x) =C \bigg\}  \nn
\ee
if and only if $\mu_0$ is concentrated on the centered sphere of radius $C^{1/p}$.
\end{lemma}

\begin{proof} 
 Let $m(x)=|x|$ be the modulus map for $x \in \R^n$, and let $\eta:= m_\#(\mu) \in \cP({[0,\infty)})$ be the push-forward of $\mu \in \cP(\R^n)$ by the map $m$. Then $\int_{\R^n} |x|^p d\mu(x) = \int_0^\infty r^p d\eta(r)$ for any $p>0$. Hence from now on we assume $\eta \in \cP({[0,\infty)})$ and $\int r^p d\eta(r) = C$. Recall Jensen's inequality, which states that
if $f:\R \to \R$ is convex and $X$ is a real-valued random variable
with average value $E[X]$, then $E[f(X)] \ge f(E[X])$, and equality holds if and only if $f$ is linear on the interval $[\inf X, \sup X]$.
With $f(r)=r^{q/p}$, Jensen's inequality yields $\int r^q d\eta(x) \ge \bigg(\int r^p d\eta(x)\bigg)^{q/p} = C^{q/p}$, and moreover equality holds if and only if $\eta$ is supported at a point in ${[0,\infty)}$, since $f$ is strictly convex on ${[0,\infty)}$. This proves the lemma. 
\end{proof} 

\begin{proof}[{Proof of Theorem \ref{main1}}] Define 
\[F(\mu)=\frac{1}{4}\iint |x-y|^4 d\mu(x)d\mu(y), \ G(\mu)=\frac{1}{2} \iint |x-y|^2 d\mu(x)d\mu(y)
\]
 so that $2E 
 =F-G$. Then for $\mu \in \cP_0(\R^n)$, 
 \[
G(\mu) = \int |x|^2 d\mu(x) 
= \Tr I(\mu) 
\]
is no longer quadratic,  but depends linearly on $\mu$ instead.
Applying the calculation from the proof of Lemma \ref{linearity}, modified slightly to account for the fact that $\int d\mu=1$ whereas $\int d\eta=0$, we get:
\begin{align}\nonumber
F(\mu) &=  
 \frac{1}{2}\int |x|^4 d\mu(x) + \frac{1}{2}
  (\Tr I(\mu))^2 
 + \Tr (I(\mu)^2).
 \end{align}
Thus the energy $\mathcal{E}_{W_{4,2}}$ is convex, and by Lemma \ref{linearity} its minimizers must all share the same second moment tensor. Convexity also implies  $\mathcal{E}_{W_{4,2}}$ admits a spherically symmetric minimizer. 
This yields that this common {second moment tensor} must be $\la {\rm Id}$ for some $\la >0$ to be determined. This leads us to define
\be
A_\la=\{ \mu \in \cP_0(\R^n) \ | \ I(\mu) = \la {\rm Id} \}. \nn
\ee
For the correct choice of $\lambda$,  $A_\la$ contains all minimizers of \eqref{energy}, and moreover by the above formulas for $F$ and $G$, for every $\mu \in A_\la$ we have 
\begin{align}\label{energyform}
2E(\mu) = \frac{1}{2}\int |x|^4 d\mu(x) + \frac{1}{2}n^2\la^2 +n \la^2 -n \la.
\end{align}
This leads us to consider minimizing the fourth moment over $A_\la$. Set
\be
B_\la=\{ \mu \in \cP_0(\R^n) \ | \ \Tr I(\mu)=n\la\}. \nn
\ee
Notice $A_\la \subset B_\la$. Now Lemma \ref{moments} asserts that $\mu$ minimizes $\int |x|^4 d\mu(x)$ over $B_\la$ if and only if $\mu$ is concentrated on the centered sphere of radius $r:=\sqrt{n\la}$. But observe that $ \sigma_r$, the uniform probability on the sphere of radius $r$, also belongs to $A_\la$. This yields that the set of minimizers $X\subset \cP_0(\R^n)$ for \eqref{energy} is precisely the following:
 \begin{align}
\nonumber X&:= \{ \mu \in \cP_0(\R^n) \cap \cP(\bbS_{\sqrt{n\la}}) \ | \ I(\mu)=\la {\rm Id} \} \\
\label{X:=}  &= \{ \mu \in \cP_0(\R^n) \cap \cP(\bbS_{\sqrt{n\la}}) \ | \ I(\mu)=c {\rm Id} \  \text{ for some } c >0 \}
 \end{align}
where $\bbS_r$ is the centered sphere of radius $r$ in $\R^n$, and the second equality is due to Lemma \ref{sameMOI}. Notice $X$ is convex since $I$ is linear in $\mu$. 

Finally let us determine the optimal $\la$. By \eqref{energyform}, $2E(\mu)=n^2\la^2+n\la^2- n\la$ for any $\mu \in X$, and $\frac{dE}{d\la}=0$ gives $\la = \frac{1}{2n+2}$, hence $r=\sqrt{n\la} = \sqrt{\frac{n}{2n+2}}$ as claimed.  
\end{proof}

\begin{example}[Infinite-dimensional attractor at transition threshold]\label{E:threshold examples}
{If $(\alpha,\beta)=(4,2)$, then the} spherical shell $\sigma_r$ of radius $r:= \sqrt{\frac{n}{2n+2}}$ is a minimizer. For others, let $\{e_i\}$ be the standard basis of $\R^n$. Then the probability $ \frac{1}{2n}\sum_{i=1}^n (\delta_{re_i} + \delta_{-re_i}) $ clearly belongs to the set $X \subset \cP_0(\R^n)$ of minimizers from \eqref{X:=}, which can be also seen by Lemma \ref{sameMOI}.
And any rotation and convex combination of these is a minimizer due to the convexity of $X$,
which shows the set of minimizers is infinite dimensional. In particular, the minimizers do not need to coincide with each other even up to rotation and translation. The uniform measure on the vertices of the regular simplex inscribed in $\bbS_r$ is also a minimizer,
by the following standard observation.
\end{example}

\begin{remark}[Second moments for the uniform measure on the vertices of a regular simplex]
\label{R:simplex moments}
Let $\nu_d \in\cP_0(\R^n)$ denote the uniform measure on the $n+1$ vertices of a regular simplex with center of mass at the origin {and diameter $d$}. Then $I(\nu_d) = \frac{d^2}{2n+2} \Id$.
\end{remark}

\begin{proof} 
 Let $\mathds{1}=(1,1,\ldots,1) \in \R^{n+1}$.  The standard simplex is
$\Delta^n := \{x \in [0,\infty)^{n+1} \mid \mathds{1}\cdot x = 1\}$. Its vertices coincide with the standard basis vectors $e_0,\ldots,e_n$ for $\R^{n+1}$. {Note that its diameter is $\sqrt2$.}
We compute the second moments $I(\nu)$ of the uniform measure $\nu=\frac1{n+1}\sum_{i=0}^n \delta_{e_i}$ over these vertices,  and its translation $T_{\la} \nu=\frac1{n+1}\sum_{i=0}^n \delta_{e_i-\la \mathds{1}}$
along the principal diagonal $\mathds{1}$ for each $\la \in \R$:
\begin{align*}
I_{jk}(T_\la \nu) 
&= \frac1{n+1}\sum_{i=0}^n
(e_i -\la \mathds{1})_j (e_i-\la \mathds{1})_k
\\ &=\frac1{n+1}({\Id}_{jk} -2 \la + (n+1)\la^2),
\end{align*}
i.e. $I(T_\la \nu) = \frac1{n+1}\Id + \la (\la - \frac2{n+1}) \mathds{1}\otimes \mathds{1}$.
Note that the choice $\la = \frac{1}{n+1}$ makes  $\nu_{\sqrt2} = T_\la \nu$ centered at the origin and lie in the subspace $\mathds{1}^\perp$, and since $I(T_\la \nu)\, v= \frac{1}{n+1}v$ for any $v \in \mathds{1}^\perp$, we have $v_i \cdot I(T_\la \nu)\, v_j = \frac{1}{n+1}{\Id}_{ij}$ for any orthonormal basis $\{v_i\}$ of $\mathds{1}^\perp$, as desired. For general diameter $d$ we multiply $(d/\sqrt2)^2$.
\end{proof} 

\begin{remark}[Concerning $d_\infty$-local energy minimizers]
For $2 < \bt<\al$ or $2\beta = 4<\al$,  
two of us showed the measure $\nu_1$ {of unit diameter} in Remark \ref{R:simplex moments}  minimizes the energy uniquely (up to rotations and translations) $d_\infty$-locally \cite{LimMcCann21}; see also Simione \cite{Simione14}.  
Example \ref{E:threshold examples} shows that for $n\ge 2$ the uniqueness part of this statement no longer holds true at the endpoint $(\bt,\al)=(2,4)$ of the latter regime,  since  $\frac12 (\nu_1 + R_\theta \nu_1)$ is also minimizing,
and lies as $d_\infty$-close to $\nu_1$ as we like when $\theta$ is sufficiently small.  
\end{remark}

\section{Identifying mildly repulsive minimizers for $\al \ge \kn $}
For $\al\bt> 0$, 
let $w_\al$ and $w_{\al,\bt}$ be defined on ${(0,\infty)}$ by
\[ w_\al(r) = \frac{r^\al}{\al}, \qquad w_{\al,\bt}(r) = \frac{r^\al}{\al} - \frac{r^\bt}{\bt},
\]
so that $W_{\al,\bt}(x) = w_{\al,\bt}(|x|)$. If $\al \ne \bt$, 
the rescaled potential
\begin{equation}\label{rescaledsymmetry} \overline{w}_{\al,\bt}(r) = \frac{w_{\al,\bt}(r) }{-w_{\al,\bt}(1)} = \frac{\bt r^\al - \al r^\bt} {\al - \bt} = \overline{w}_{\bt,\al}(r)
\end{equation}
then satisfies $\overline{w}_{\al,\bt}(r) \ge -1$ 
on $r \ge 0$ with equality if and only if $r=1$. {We note that, while the present work is concerned only with the case where $\alpha>\beta>0,$ the rescaled potential $\overline{w}_{\alpha\beta}$ continues to satisfy $\overline{w}_{\alpha,\beta}\ge 1$ with equality if and only if $r=1$ on the broader regime $\alpha\beta>0.$
If instead $\alpha\beta<0$, then $\overline{w}_{\alpha,\beta}$ is uniquely maximized at $r=1.$} Define $\overline{W}_{\al,\bt}(x)=  \overline{w}_{\al,\bt}(|x|)$. Obviously  $ \mathcal{E}_{W_{\al,\bt}}$ and  $ \mathcal{E}_{\overline{W}_{\al,\bt}}$ share the same minimizers on $\cP(\R^n)$ as long as $\al>\bt$.
The crux of the proof of Theorem \ref{main2} is the following monotonicity:

\begin{lemma}[Rescaled potential increases with either exponent]
\label{positive}
For each $\al \ne 0$, $\bt \ne \al$, $r>0$, we have $\al\frac{\p}{\p \bt}\overline{w}_{\al,\bt}(r) \ge 0$ with equality holding if and only if $r=1$.
\end{lemma}
\begin{proof} 
Direct computation yields
\be
\al\frac{\p}{\p \bt}\overline{w}_{\al,\bt}(r) = \frac{\al^2 r^\bt}{(\al-\bt)^2}(r^{\al-\bt} -1 - \log r^{\al-\bt}). \nn
\ee
From this, the lemma follows from the fact that the function $t \mapsto t -1 -\log t \ge 0$ for $t > 0$ with equality holding only if $t=1$.
\end{proof}

\begin{figure}[H]
    \centering
    \includegraphics[scale=0.25]{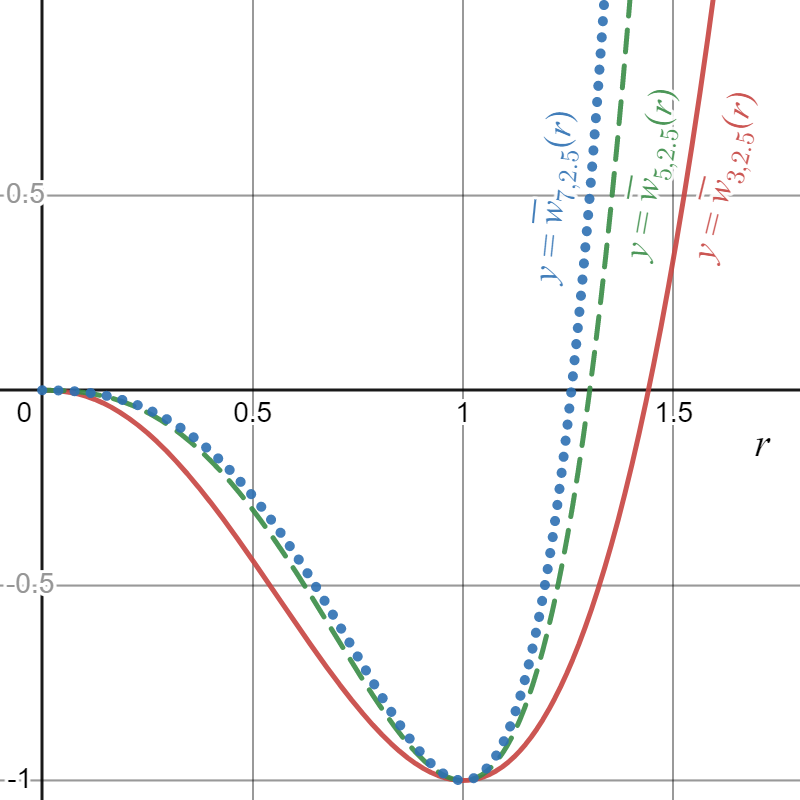}
    \caption{{Graphs of $\overline{w}_{\alpha,\beta}(r)$ for $\beta=2.5$ and $\alpha=3,5,7$ illustrating the results of Lemma \ref{positive}. In particular, although all three graphs agree for $r\in\{0,1\},$ we see that $\overline{w}_{\alpha,\beta}(r)$ is a strictly increasing function of $\alpha$ for each $r\in (0,1)\cup(1,\infty).$ Note that the symmetry 
    $\overline{w}_{\al,\bt}=\overline{w}_{\bt,\al}$ from \eqref{rescaledsymmetry} ensures that the monotonicity in $\beta$ proven in Lemma \ref{positive} and the monotonicity in $\alpha$ shown in this figure are equivalent.}}
    \label{fig:2d}
\end{figure}

\begin{proof}[{Proof of Theorem \ref{main2}}] Assume $\al > \bt > 0$ and $\cP_{\Delta^n}$ minimizes $ \mathcal{E}_{W_{\al,\bt}}$.  It is enough to prove $\cP_{\Delta^n}$ uniquely minimizes both $\mathcal{E}_{W_{\al+\epsilon,\bt}}$ and 
$ \mathcal{E}_{W_{\al,\bt+\epsilon}}$ on $\cP(\Rn)$ for all $\ep \in (0, \al-\bt)$, and that the support of $\nu \in \cP_{\De^n}$ uniquely minimizes both
$\nu * W_{\al+\ep,\bt}$ and $\nu * W_{\al,\bt+\ep}$ on $\Rn$.
Let $\rho(x,y)=|x-y|$. For $\mu \in \cP(\R^n)$, observe the push-forward $ \tilde \mu:= \rho_\#(\mu \otimes \mu) \in \cP({[0,\infty)})$ via the map $\rho$ satisfies, since $W(x) = w(|x|)$,
\be
 \mathcal{E}_{\overline{W}_{\al,\bt}}(\mu) =\frac12 \int_0^\infty \overline{w}_{\al,\bt}(r) d\tilde\mu(r).
\ee
Let $\nu \in \cP_{\Delta^n}$. By assumption $\int \overline{w}_{\al,\bt}(r) d\tilde\mu(r) \ge \int \overline{w}_{\al,\bt}(r) d\tilde\nu(r)$. 
Since $\spt(\tilde \nu) = \{0,1\}$ and $\overline{w}_{\al,\bt}(r)$ is constant in $\al>\bt>0$ at $r=0$ and $1$,
\begin{align*}
 \int \overline{w}_{\al,\bt}(r) d\tilde\nu(r) =  \int \overline{w}_{\al+\ep,\bt}(r) d\tilde\nu(r)
 =  \int \overline{w}_{\al,\bt+\ep}(r) d\tilde\nu(r) 
 \end{align*}
 for all $0<\ep < \al-\bt$.
On the other hand, by Lemma \ref{positive} 
(and the symmetry of $\overline{w}$ in $\al,\bt$), $\ep>0$ implies
\begin{align*}
 \int \overline{w}_{\al,\bt}(r) d\tilde\mu(r) \le  \int \overline{w}_{\al+\ep,\bt}(r) d\tilde\mu(r), \  \int \overline{w}_{\al,\bt}(r) d\tilde\mu(r) \le \int \overline{w}_{\al,\bt+\ep}(r) d\tilde\mu(r)
\end{align*}
with equality holding only if $\spt(\tilde \mu) \subset \{0,1\}$, i.e. only if $\mu$ is concentrated on the set of vertices of a unit simplex. {Hence,} if $\mu$ minimizes $\mathcal{E}_{\overline{W}_{\al+\ep,\bt}}$ or $\mathcal{E}_{\overline{W}_{\al,\bt+\ep}}$ 
{then it must be concentrated on the vertices of a unit simplex. Thus, we use the isometry described in \cite[Remark 1.2]{LimMcCann21} to write $\mu=\sum_{i=1}^{n+1} m_i\delta_{2^{-1/2}e_i}\in \mathcal{P}(\R^n).$ In particular, if we let $m=(m_1,...,m_{n+1})$ be the vector of masses and if we, without loss of generality, define an $(n+1)\times(n+1)$ matrix $A$  by $$A_{ij}=n\overline{w}_{\alpha+\epsilon,\beta}\left(\frac{1}{\sqrt{2}}\right)\left[\frac{1-\operatorname{Id}_{ij}}{n}\right],$$ then we may write 
$$\mathcal{E}_{\overline{W}_{\alpha+\epsilon, \beta}}(\mu)=m^T A m=n\overline{w}_{\alpha+\epsilon,\beta}\left(\frac{1}{\sqrt{2}}\right)m^T\overline{A}m,$$
where we define $\overline{A}_{ij}:=\frac{1}{n}(1-\operatorname{Id}_{ij}).$
Thus, noting that each of the rows and columns of $\overline{A}$ sums to $1$, and noting that $\overline{A}^2$ is a positive matrix, the Perron-Frobenius theorem implies that $\overline{A}$ has $1$ as an eigenvalue with multiplicity $1$, and that all other eigenvalues of $\overline{A}$ have absolute value less than $1.$ Since $\overline{m}:=\frac{1}{n+1}(1,...,1)$ is an eigenvector of $\overline{A}$ with eigenvalue $1,$ the spectral theorem implies that $m$ maximizes the quantity $m^T\overline{A}m$ among all vectors in $\R^{n+1}$ with entries summing to $1.$ In turn, since the constant $\overline{w}_{\alpha+\varepsilon,\beta}(\frac{1}{\sqrt{2}})$ is negative, we conclude that, in order to minimize $\mathcal{E}_{\overline{W}_{\alpha+\epsilon, \beta}},$} 
$\mu$ must uniformly distribute its mass over the vertices of a unit simplex, i.e. $\mu \in \cP_{\Delta^n}$. 
This proves the first identity \eqref{twin argmin identities}.

Observe that the Euler-Lagrange equation from e.g. \cite{DaviesLimMcCann21+a} asserts
\begin{equation}\label{EL}
\spt \nu \subset \argmin (\nu * W_{\al,\bt}).
\end{equation}
Since the vertices of a unit simplex, $\spt \nu$, is characterized as the maximal set of points at distance one from each other, 
Lemma \ref{positive} shows
\[
\nu * W_{\al,\bt} \le \nu * W_{\al+\ep,\bt} \quad{\rm and}\quad
\nu * W_{\al,\bt} \le \nu * W_{\al,\bt+\ep}\]
with equalities holding precisely on $\spt \nu$.  This implies the second identity \eqref{twin argmin identities} to establish Theorem \ref{main2}.
\end{proof}

\begin{proof}[{Proof of Corollary \ref{C:main4}}]
Lemma \ref{positive} shows $\overline{w}_{\al,\bt}(r)$ is a nondecreasing function of $\bt \in (0,\al)$, 
and strictly increasing unless $r \in \{0,1\}$.
Also $\lim_{\bt \to \al} \overline{w}_{\al,\bt}(r) =  r^\al (\al\log r - 1)$, so
 $\lim_{\bt \to \al} \overline{W}_{\al,\bt}(x) =  D_\al (x)$.
As in the preceding proof,  if $\cP_{\Delta^n}$ minimizes ${\mathcal E}_{W_{\al,\bt}}$, comparison shows it then minimizes
 ${\mathcal E}_{D_\al}$ uniquely. Conversely if $\cP_{\Delta^n}$ minimizes ${\mathcal E}_{D_\bt}$, 
 then minimizes ${\mathcal E}_{W_{\al,\bt}}$ uniquely for all $\al>\bt$.
 \end{proof}

\begin{proof}[Proof of Corollary \ref{C:main3}]
Theorems \ref{main1}--\ref{main2} and Remarks \ref{R:1D} and \ref{R:simplex moments} yield Corollary \ref{C:main3}. 
\end{proof}

\section{The transition threshold}
\label{S:threshold}
  In this section,  we establish the existence of a transition threshold $\al_{\Delta^n}(\bt)$ which separates the part of the mildly repulsive region $\bt \ge 2$
on which equidistribution  $\cP_{\Delta^n}$ over the vertices of the unit simplex minimizes the energy $\cE_{W_{\al,\bt}}$ from the part on which it does not.
Above the threshold,  these minimizers are unique up to rigid motions.
We also establish that this threshold lies in the range $[\alnd(\bt),\ali(\bt)] \subset [\aln(\bt),\ali(\bt)]$ given by Definitions \ref{upperbounddef}, \ref{D:threshold lower bound} 
and \ref{D:weaker lower bound}, 
which collapses to the point $\{\ali(\bt)\}$ in the high dimensional limit (Proposition \ref{P:high dimensional collapse}).


\begin{proof} [{Proof of Theorem \ref{T:threshold}}]
For $\bt\ge2$, the existence of $\al_{\Delta^n}(\bt) \in [\bt,\infty]$ satisfying \eqref{gthreshold} and \eqref{lthreshold} follow 
from Theorem \ref{main2};
 also $\al_{\Delta^n}(\bt) <\infty$ is asserted in \cite{LimMcCann21}.
The fact that $\al_{\Delta^n}(2) \le\kn$, 
existence of a minimal $\beta_n \in [2,\kn ]$ such that $\al_{\Delta^n}(\bt) = \bt$ for $\bt > \bt_n$, and (nonstrict) monotonicity of $\al_{\Delta^n}:[2, \bt_n] \longrightarrow [\bt_n, \kn]$ are consequences of Corollary \ref{C:main3}.
The centrifugal value $\al_{\De^n}(2)=\kn $ follows from Theorem \ref{main1} and Remark \ref{R:1D}.
We next establish that at least one of the containments \eqref{ethreshold} is strict by combining results from \cite{LimMcCann21}
with the strategy used to provide an analogous statement for a related problem in \cite{LimMcCann21+}.

For $p \in [1,\infty]$, recall that the Kantorovich-Rubinstein-Wasserstein distance between $\mu,\mu' \in \cP(\Rn)$ is defined by
\be\label{KRW metric}
d_p(\mu,\mu') 
:=\inf_{X \sim \mu, Y \sim \mu'} \| X - Y\|_{L^p},
\end{equation}
where the infimum is taken over arbitrary couplings of random vectors $X$ and $Y$ in $\Rn$
whose laws are given by $\mu$ and $\mu'$ respectively. 
The metrics $d_p$ are well-known to metrize weak convergence of measures on compact subsets $K \subset\Rn$
unless $p=\infty$ \cite{Villani03}.
Given such a compact set $K \subset \Rn$ and $\al > \bt \ge 2$,  we first claim that if $(\al,\bt) = \lim_{k\to\infty}(\al(k),\bt(k))$ for a sequence $\al(k)>\bt(k) \ge 2$,  
then the functionals $\cE_{W_{\al(k),\bt(k)}}$ $\Gamma$-converge to $\cE_{\al,\bt}$ on $(\Prob(K),d_2)$.
Since the potentials $\{W_{\al(k),\bt(k)}\}_k$ are uniformly equicontinuous on $K \times K$,  this is easy to prove 
using the argument, e.g., from Lemma 3.2 of \cite{LimMcCann21+}, so we do not give more details here.  Now Proposition 1.1 of \cite{DaviesLimMcCann21+a} ensures the minimizers of $\cE_{W_{\al,\bt}}$ on 
$\Prob(\Rn)$ exist and 
can all be translated to lie in a centered ball of radius $e^{1/\bt}$; as $k\to \infty$
it follows from this $\Gamma$-convergence that $d_2$-accumulation points of minimizers of $\cE_{\al(k),\bt(k)}$ 
therefore minimize $\cE_{\al,\bt}$ on $\cP(\Rn)$.  Taking $\bt(k)=\bt$ and $\al(k)\searrow \al_{\Delta^n}(\bt)$ 
then shows that the (nonstrict) first containment of \eqref{ethreshold} is a consequence of \eqref{gthreshold}.  
When $\al_{\De^n}(\bt)=\bt$, strict containment becomes trivial. We may therefore assume $ \al_{\De^n}(\bt)=: \al >\bt$,
and let $\bt(k)=\bt$ and $\al(k) \nearrow \al$.
We also assume $\bt>2$ because for $\bt=2 \le n$ strict containment follows from Theorem \ref{main1},  while for $(\bt,n) = (2,1)$ it is easy to check $\spt \psi = \{-\frac12, \frac12\} \subsetneq [-\frac12, \frac12]= \argmin (W_{3,2} * \psi)$. Since 
there exist minimizers $\mu_k$ of $\cE_{\al(k),\bt}$ on $\cP(\Rn)$ whose
support lies in the centered ball of radius $e^{1/\beta}$, weak compactness of the probability measures on this ball
yields a subsequential limit $d_2(\mu_k,\mu_\infty) \to 0$ (the subsequence having been relabelled $\mu_k$);  $\Gamma$-convergence
then ensures $\mu_\infty$ minimizes $\cE_{W_{\al,\bt}}$ on $\cP\Big(\overline{B_{e^{1/\beta}}(0)}\Big)$,  
hence on $\cP(\Rn)$ by \cite[Proposition 2.1]{DaviesLimMcCann21+a}.

The second containment in \eqref{ethreshold} follows from the first and the Euler-Lagrange equation described e.g. in Proposition 1.1 
of \cite{DaviesLimMcCann21+a}. To derive a contradiction,  assume neither containment in \eqref{ethreshold} is strict, 
so that $\mu_\infty \in \cP_{\De^n}$
and 
\be\label{to be proved}
\spt \mu_\infty = \argmin_{\Rn} W_{\al,\bt} * \mu_\infty.
\ee
 Set $\spt \mu_\infty = \{x_0,\ldots,x_n\}$ and $0<R<1/2$.
Since $d_2(\mu_k,\mu_\infty) \to 0$ and the  Euler-Lagrange equation applied to $\mu_k$, 
and the uniform convergence on every ball of  $W_{\al(k),\bt} * \mu_k$ to $W_{\al,\bt} * \mu_\infty$  together with  \eqref{to be proved} yields 
$$
1=\mu_k[\cup_{i=0}^n B_R(x_i)], \q {\rm while} \q \mu_k[B_R(x_i)] \in (\frac1{n+2},\frac1{n})
$$
for $k$ sufficiently large;   c.f. Lemma 4.3 of \cite{LimMcCann21+} or Corollary 3.6 of \cite{LimMcCann21}.
Setting 
\be\label{concentrated comparator}
\mu_k' := \sum_{i=0}^n \mu_k[B_R(x_i)] \delta_{x_i}
\ee
ensures $d_\infty(\mu_k,\mu_k')<R$.  On the other hand, if $\al(k)> \bt^*:=\frac13(\al+2\bt)$,
Corollary 4.3 of \cite{LimMcCann21} provides $r=r(\bt,\bt^*,n)$ such that 
$\mu_k'$ (and its rotates and translates) uniquely minimize $\cE_{W_{\al(k),\bt}}$ on a $d_\infty$-ball
of radius $r$ around $\mu_k'$.  But $\mu_k$ was chosen to minimize $\cE_{W_{\al(k),\bt}}$ globally on $\cP(\Rn)$.
Taking $R<r$ and $k$ correspondingly large therefore forces $\mu_k$ to be a rotate or translate of $\mu_{k}'$.
From e.g. the Perron-Frobenius theorem, $\mu_k$ then assigns equal mass to each point in $\spt\mu_k$,
hence $\mu_k \in \cP_{\De^n}$. Since $\al(k)<\al_{\De^n}(\bt)$ by construction, \eqref{lthreshold} produces the desired contradiction  $\mu_k \not\in \cP_{\De^n}$,
to establish that at least one of the containments in \eqref{ethreshold} is strict. From this, notice the monotonicity of $\al_{\Delta^n}:[2, \bt_n] \longrightarrow [\bt_n, \kn]$ must be strict in view of \eqref{twin argmin identities}, and implies $\bt_n \in (2,\kn )$.

It remains to deduce continuity of $\al_{\De^n}$ at each $\bt \in [2,\bt_n]$.
 Set
\[
\al_{\De^n}(\bt\pm) := \lim_{\ep \downarrow 0}  \al_{\De^n}(\bt\pm\ep).
\]
If $\al \in (\al_{\De^n}(\bt), \al_{\De^n}(\bt-))$ for some $\bt \in (2,\bt_n]$, then choosing $\mu_k$ to minimize $\cE_{W_{\al,\bt-1/k}}$ 
on $\cP(\Rn)$, after translation into a centered ball of radius $e^{1/{(\bt-1)}}$
 we can extract a subsequential $d_2$-limit $\mu_\infty$ of $\mu_k$. Notice  $\mu_k \not\in \cP_{\De^n}$,  while
$\Gamma$-convergence implies  $\mu_\infty$ minimizes $\cE_{\al,\bt}$ 
hence $\mu_\infty \in \cP_{\De^n}$ by Theorem \ref{main2}.  
 But then as above,  this contradicts the $d_\infty$-unique local minimality of $\mu_k'$ from \eqref{concentrated comparator} for $R$ sufficiently small and 
$k$ correspondingly large.
On the other hand, if $\al \in (\al_{\De^n}(\bt+), \al_{\De^n}(\bt))$ for some $\bt \in [2,\bt_n]$, then choosing $\mu_k$ to minimize $\cE_{W_{\al,\bt+1/k}}$ 
on $\cP(\Rn)$, we can extract a subsequential $d_2$-limit $\mu_\infty$ of $\mu_k$.  This time  $\mu_k \in \cP_{\De^n}$, while
$\Gamma$-convergence and $\al<\al_{\De^n}(\bt)$ imply $\mu_\infty \not\in \cP_{\De^n}$,  contradicting the fact that $\cP_{\De^n}$ is $d_2$-closed.
We conclude the desired continuity $\al_{\De^n}(\bt)= \al_{\De^n}(\bt\pm)$, which also implies $\al_{\De^n}(\bt_n)=\bt_n$.
\end{proof}

\subsection{Threshold upper bound independent of dimension $n \ge 2$}\label{SS:upper}

We now establish an upper bound $\ali(\bt)$ for the threshold $\alpha_{\Delta^n} (\bt)$.
Note that this upper bound and the quantities $\bti$ and $\fin(\bt)$ defining it become independent of dimension as soon as $n \ge 2$.
The asterisk on these quantities reminds us of their implicit dependence on $\min\{n,2\}$, however.

\begin{definition}[Threshold upper bound]\label{upperbounddef}
Set
\[
\bti:= \frac{\kn-2}{\log (\kn/2)} = 
\begin{cases} 
\frac{1}{\log(3/2)} &\ {\rm if}\ n=1,
\\ \frac{2}{\log 2} &\ {\rm if}\ n\ge 2.
\end{cases}
\]
For $\bt \ge 2$, 
define $\ali=\ali(\beta)$ as the largest solution of 
\begin{equation}\label{upperdefining}
\frac{e^{\alpha/\bti}}{\alpha}=\frac{e^{\beta/\bti}}{\beta}.
\end{equation}
\end{definition} 

\begin{remark}[Number of solutions]\label{R:number of solutions}
For any given $\beta\ge 2$ and $n\in\{1,2\}$, there are at most two solutions to equation (\ref{upperdefining}), which follows from the fact that $\fin(t):=1 -\frac{e^{t/\bti}}{t}$ is unimodal on $(0,\infty),$ 
 i.e. has a unique global maximum and no local minima.
In particular, we see 
$
t^2\bti e^{-t/\bti}\frac{d \fin}{dt} 
=t-\bti
$
is positive on $(0,\bti),$ zero at $\bti$, and negative on $(\bti,\infty)$.
Thus $\ali(\bt)=\bt$ if and only if $\bt \ge \bti$.
\end{remark}

\begin{remark}[Alternative interpretation]\label{alternative} 
Set
\[
\bar{w}_{\bt,\bt}(r) :=\lim_{\al \to \bt} \bar{w}_{\al,\bt}(r) = r^\bt (\bt\log r - 1),
\]
and let $z_{\al,\bt}$
denote the positive zero of $\bar{w}_{\al,\bt}$, where $z_{\al,\bt}=(\frac{\al}{\bt})^{\frac{1}{\al-\bt}}$ for $\al\ne \bt$ and  $z_{\bt,\bt}:= e^{1/\bt}$. Notice that $z_{\kn ,2}=\frac{3}{2},$ if $n=1$, and $z_{\kn,2}=\sqrt{2},$ if $n\ge 2.$ Hence, after some rearranging, we obtain $\bti$ from the equation $z_{\bti,\bti}=z_{\kn,2}$ and $\ali$ as the largest solution of $z_{\alpha,\beta}=z_{\kn,2},$ or rather $w_{\alpha,\beta}(z_{\kn,2})=0$.
\end{remark}

The following lemma and corollary demonstrate that $\ali$ is indeed an upper bound for the threshold function:
\begin{lemma}[Comparing pair potentials]
\label{signchange}
Let $2<\beta<\alpha<\kn$. Then
$\bar{w}_{\kn,2}(r)\le \bar{w}_{\alpha,\beta}(r)$ for all $r\in [0,z_{\alpha,\beta}]$ if and only if $z_{\alpha,\beta}\le z_{\kn,2}$.
\end{lemma}

\begin{proof}
One direction is trivial, as if $\bar{w}_{\kn,2}(r)\le \bar{w}_{\alpha,\beta}(r)$ for all $r\in [0,z_{\alpha,\beta}]$, then in particular $\bar{w}_{\kn,2}(z_{\alpha,\beta})\le \bar{w}_{\alpha,\beta}(z_{\alpha,\beta})=0,$ hence $z_{\alpha,\beta}\le z_{\kn,2}$. For the proof of the other direction, we begin by defining 
\[
g(r):=\bar{w}_{\kn,2}(r)-\bar{w}_{\alpha,\beta}(r)=\frac{2r^{\kn}-\kn r^2}{\kn-2}-\frac{\beta r^\alpha-\alpha r^\beta}{\alpha-\beta}.
\]
We may divine the behaviour of $g$ from its fifth derivative
\[
g^{(5)}(r)=\frac{\alpha\beta r^{\beta-5}}{\alpha-\beta}\left[-r^{\alpha-\beta}\prod_{i=1}^4(\alpha-i)+\prod_{i=1}^4(\beta-i)\right]
\]
for $r\in (0,\infty).$
Written in this form, we see that $g^{(5)}(r)$ is {the product of a positive function of $r$ and a} monotone {function of $r$,} and hence has at most one sign change. {More precisely}, $g^{(3)}(r)$ is either convex-concave {(if $\alpha<3$)}, concave-convex {(if $\beta>3$)}, or strictly convex {(if $\beta\le 3\le \alpha\ne \bt$)} on $(0,\infty).$ Moreover, we may write 
\[
g^{(3)}(r)=2\cdot \kn(\kn-1)r^{\kn-3}+\frac{\alpha\beta}{\alpha-\beta}\left[-(\alpha-1)(\alpha-2) r^{\alpha-3}+(\beta-1)(\beta-2)r^{\beta-3}\right].
\]
Here, both the highest order term $r^{\kn-3}$ and the lowest order term $r^{\beta-3}$ have positive coefficients, which implies that $g^{(3)}$ is positive outside a compact subinterval of $(0,\infty)$. This, combined with the convex/concave structure of $g^{(3)}$, implies that $g^{(3)}$ can have at most two zeros on $(0,\infty)$ and, in particular, may change signs at most twice --- from positive to negative to positive. 

This implies either $g'$ is convex-concave-convex on $(0,\infty)$ or just convex. We may assume $g'$ is convex-concave-convex as, if it is simply convex, an easier argument than what follows will yield the desired conclusion. Notice that 

$$g'(r)=\frac{\kn\cdot2}{\kn-2}(r^{\kn-1}-r)-\frac{\alpha\beta}{\alpha-\beta}(r^{\alpha-1}-r^{\beta-1})$$
is negative near zero and hence, the convex-concave-convexity implies $g'$ changes sign at most thrice on $(0,\infty).$ Note $g'(0)=g'(1)=0 = g(0)=g(1)$. Since $g'$ is negative near zero, we see that $g'$ must change from negative to positive somewhere in $(0,1)$, implying the existence of a zero of $g'$ on this interval. Hence $g'$ has at most one zero on $(1,\infty)$. But if there is no zero on $(1,\infty)$, then the shape of $g'$ and $g(1)=g'(1)=0$ implies $g' >0$ hence $g > 0$ on $(1, \infty)$, yielding $z_{\alpha,\beta}>z_{\kn, 2}$, a contradiction. Hence we deduce that, on $(1, \infty)$, $g'$ changes sign from negative to positive. With $g(1)=0$, this implies $g$ also changes sign from negative to positive on $(1,\infty)$. Now since the condition $z_{\alpha,\beta}\le z_{\kn,2}$ clearly implies $g(z_{\alpha, \beta})\le 0$, this allows us to conclude that $g\le 0$ on $[1, z_{\alpha,\beta}].$ 

It remains to show $g\le 0$ on $[0,1]$. Assume $g$ is positive somewhere in $ (0,1)$. Then $g'$ would have to change signs (at least) twice on the interval $(0,1)$, from negative to positive to negative. With $g'(1)=0$ all three zeros of $g'$ are in $(0,1]$, thus no zero on $(1,\infty)$, contradiction as before. This concludes the proof.
\end{proof}

\begin{figure}[H]
\caption{{Comparison of $\overline{w}_{4,2}$ to $\overline{w}_{3.1,2.5}$ and $\overline{w}_{3.5,2.5}$ illustrates Lemma \ref{signchange}. Note that the unique positive root $z_{\al,\bt}$ of $\overline{w}_{\al,\bt}$ is ordered so that
$z_{3.5,2.5}=1.4<z_{4,2}=\sqrt{2}<z_{3.1,2.5}=\left(\frac{3.1}{2.5}\right)^{\frac{1}{3.1-2.5}}\approx 1.431$. On one hand, the proof of the lemma 
implies that $\overline{w}_{3.1,2.5}(z_{3.1,2.5})=0<w_{4,2}(z_{3.1,2.5})$. Thus, by continuity, there exists some $\varepsilon>0$ such that $\overline{w}_{3.1,2.5}(r)<\overline{w}_{4,2}(r)$ for all $r\in (z_{3.1,2.5}-\varepsilon, z_{3.1,2.5}]\approx (0.7332, 1.415]$. On the other hand, 
the lemma implies that, since the graph of $\overline{w}_{3.5, 2.5}$ intersects the $x$-axis at $z_{3.5, 2.5}<\sqrt{2},$ the inequality $\overline{w}_{3.5, 2.5}(z_{3.5, 2.5})\ge \overline{w}_{4, 2}(z_{3.5, 2.5})$ extends to all $r\in (0, z_{3.5, 2.5}].$}}
\begin{subfigure}{.5\textwidth}
  \centering
  \includegraphics[width=.8\linewidth]{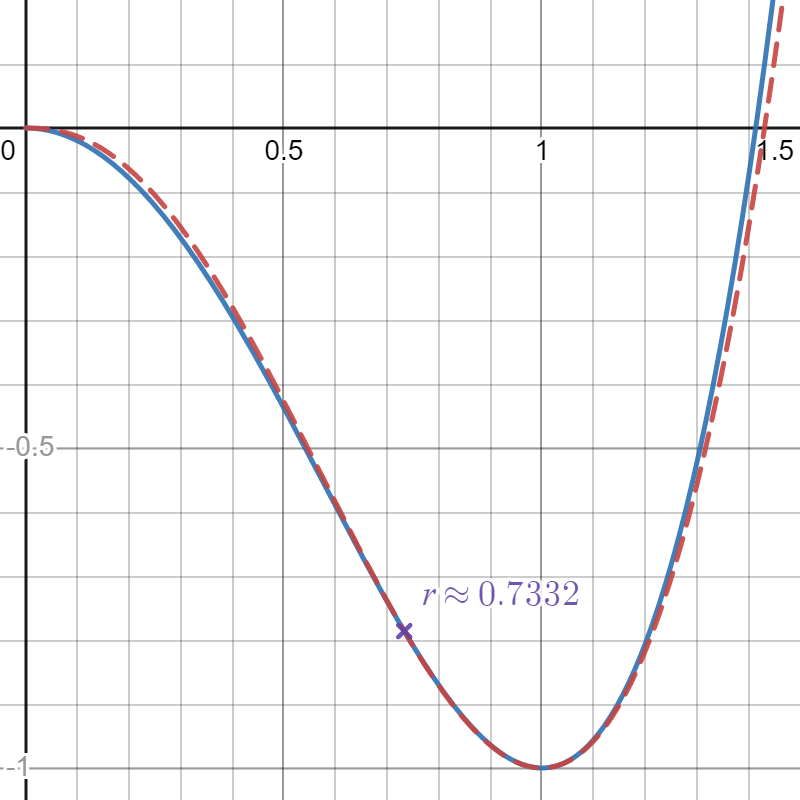}
  \caption{{Comparing $\overline{w}_{4,2}$ and $\overline{w}_{3.1,2.5}$}}
  \label{fig:sfig1}
\end{subfigure}%
\begin{subfigure}{.5\textwidth}
  \centering
  \includegraphics[width=.8\linewidth]{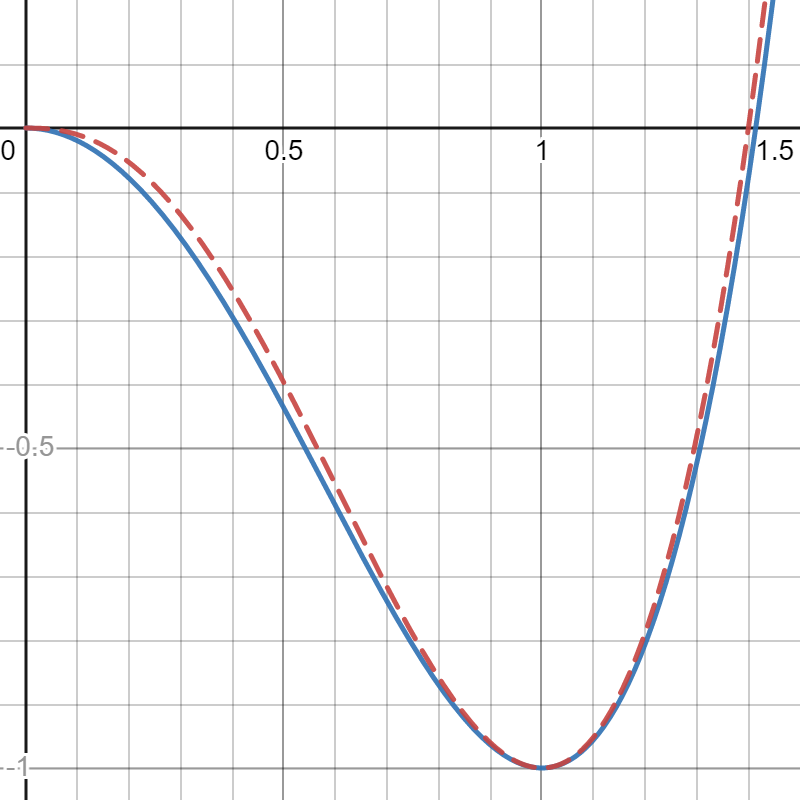}
  \caption{{Comparing $\overline{w}_{4,2}$ and $\overline{w}_{3.5,2.5}$}}
  \label{fig:sfig2}
\end{subfigure}
\label{fig:4.5}
\end{figure}

\begin{corollary}[Threshold upper bound]
\label{upperbound} If $\bt \ge 2$ then $ \al_{\Delta^n}(\bt) \le \ali(\bt)$.
\end{corollary}

\begin{proof} Recall $\cP_{\Delta^n}$ minimizes $\cE_{W_{\kn,2}}$ from Corollary \ref{C:main3}.
 The fact {from Lemma \ref{signchange}, namely} $\bar{w}_{\kn, 2} (r) \le \bar{w}_{\ali, \bt}(r)$ on $r\in[0, z_{\ali, \bt}]$ with equality at $r=1$,  shows $\cP_{\Delta^n}$ minimizes $\cE_{W_{\ali, \bt}}$, since any minimizer of $\cE_{W_{\ali, \bt}}$ has its diameter no greater than $z_{\ali, \bt}$, by \cite[Lemma 1]{KKLS21}.
\end{proof}

\subsection{Threshold lower bound for each dimension}\label{SS:Lower}

We now derive a dimension dependent lower bound $\alnd$ for $\alpha_{\Delta^n}$ from the Euler-Lagrange equation (\ref{EL}) for minimizers. 

\begin{definition}[Threshold lower bound] \label{D:threshold lower bound}
Let $ \nu  \in \cP_{\Delta^n}$. For each $\bt \ge 2$, define $\alnd(\bt) \in [\bt, \infty)$ to be
\begin{align}\label{threshold lower bound}
\alnd(\bt)&:= \inf\{ \al > \bt \ | \ \spt \nu \subset \argmin_{\Rn} (W_{\al,\bt} * \nu) \}  \\
& = \sup \{ \al \in \R \ | \ \spt \nu \not\subset \argmin_{\Rn} (W_{\al,\bt} * \nu)    \}. \nn
\end{align}
\end{definition}

\begin{proposition}[Threshold lower bound]\label{P:threshold lower bound}
\label{lowerbound}
Let $ \nu  \in \cP_{\Delta^n}$.
 If $\al>\alnd(\beta)$ for some $\bt \ge 2$, then $\spt\nu =\argmin_{\R^n}(W_{\al,\bt}*\nu)$.
In particular,
\be\label{anotherdef}
\alnd(\bt)= \inf\{ \al > \bt \ | \ \spt \nu = \argmin_{\Rn} (W_{\al,\bt} * \nu) \},
\ee
and $\alnd\le \alpha_{\Delta^n}$.
\end{proposition}

\begin{proof}
For any $\al > \alnd(\bt)$, notice Lemma \ref{positive} yields $\spt \nu = \argmin_{\Rn} (W_{\al,\bt} * \nu) $, which gives \eqref{anotherdef}. The fact that 
$\alnd \le \alpha_{\Delta^n}$ is a direct consequence of the Euler-Lagrange equation satisfied by a minimizer: i.e. if $\al \ge \al_{\Delta^n}$, so that $\nu \in \cP_{\Delta^n}$ minimizes $\cE_{W_{\al,\bt}}$, then $\nu$ satisfies \eqref{EL} hence $\al \ge \alnd$.
\end{proof}

Although the value of $\alnd(\beta)$ is not very explicit,  it is possible to estimate it explicitly from below by evaluating the potential $W_{\al,\bt} * \nu$ 
at points chosen judiciously to expose potential violations of the Euler-Lagrange equation.
The resulting estimates $\aln \le \alnd$ provide weaker but explicit lower bounds for the threshold.
This requires the following family of functions and their unimodality:

{\begin{definition}[A family of unimodal functions]
Define $f_n:(0,\infty)\to\R$ by
\begin{equation}\label{f_n}
f_n(t):=\begin{cases}
\frac{2^{-1}-2^{-t}}{t} & \text{ if }n=1\\
\frac{n-(\frac{2n}{n+1})^{t/2}-n(\frac{n-1}{n+1})^{t/2}}{t} & \text{ if }n\ge 2.
\end{cases}
\end{equation}
\end{definition}

Using this family of functions, we define a new family of lower bounds:

\begin{definition}[A weaker threshold lower bound]\label{D:weaker lower bound}
For $\beta\ge 2,$ define $\aln(\beta)$ by 
\begin{equation}\label{threshold lower estimate}
\aln(\beta)=\max\{\alpha\ge 2 \ | \  f_n(\al)=f_n(\beta)\}.
\end{equation}
\end{definition}
In particular, the set over which we take the maximum in the previous definition has at most two elements, as the following lemma shows:

\begin{lemma}[Unimodality of $f_n$]\label{unimodality}
 For any $n\ge 1$, the function $f_n(t)$ is unimodal on $t \in (0,\infty)$. 
Indeed, $f_n$ admits a unique global maximum $\underline{\beta}_n:=\argmax_{t>0}f_n(t)$ and 
 no other critical points.
\end{lemma}

\begin{proof}
We first treat the case $n=1$ separately. Here, notice that $f_1'(t)$ has the same sign as $g_1(t):=t^2 f_1'(t)=(t\log 2+1)2^{-t}-2^{-1}$.
Since $g_1'(t)=-t2^{-t}\log^2 2$ is always negative, and since $g_1(0)=\frac{1}{2}$ and $\lim_{t\to\infty}g_1(t)=-\frac{1}{2},$ we conclude that $f_1'$ switches sign from positive to negative at its unique zero in $(0,\infty),$ and has no other sign changes.  We denote the unique zero of $f_1'$ by $\underline{\beta}_1.$

The $n\ge 2$ case proceeds in a similar manner. Here, we notice that 
\begin{align*}
g_n(t)
:=t^2 f_n'(t)
=-\frac{t}{2} &\left[\left(\frac{2n}{n+1}\right)^{t/2}\log \frac{2n}{n+1}+n\left(\frac{n-1}{n+1}\right)^{t/2}\log \frac{n-1}{n+1}\right]
\\&-n+\left(\frac{2n}{n+1}\right)^{t/2}+n\left(\frac{n-1}{n+1}\right)^{t/2},\end{align*}
and compute 
\[
g_n'(t)=-\frac{t}{4}\left[\left(\frac{2n}{n+1}\right)^{t/2}\log^2 \frac{2n}{n+1}+n\left(\frac{n-1}{n+1}\right)^{t/2}\log^2\frac{n-1}{n+1}\right].
\]
Since $g_n'(t)$ is negative everywhere, $g_n(0)=1,$ and $\lim_{t\to\infty}g_n(t)= -\infty$, we may apply an identical argument to the one employed in the $n=1$ case to show the existence of $\underline{\beta}_n$ with all desired properties.
 \end{proof}

\begin{remark}[Diagonal intersects bound]
Notice $\aln(\beta)>\beta$ if and only if $\beta<\underline{\beta}_n$. That is, the graph of $\underline{\al}_{\De^n}$ intersects the line $\alpha=\beta$ at the point $(\underline{\beta}_n,\underline{\beta}_n)$. 
\end{remark}

\begin{proposition}[Estimating threshold lower bound]\label{lowerlowerboundprop}
For $\beta\ge 2$, the thresholds of Definitions \ref{D:weaker lower bound}, \ref{D:threshold lower bound} and
Theorem \ref{T:threshold} 
satisfy
$\aln(\beta)\le  \alnd(\beta)\le 
\al_{\Delta^n}(\beta) 
$.
\end{proposition}

\begin{proof}
In view of Proposition \ref{lowerbound} we need only show $\aln(\beta)\le\alnd(\bt)$.
We proceed by relating the defining equations for $\aln$ to the Euler-Lagrange equation for a unit simplex $\nu\in \cP_{\Delta^n}$. As in the introduction, we denote the vertices of the unit $n$-simplex by $\{x_0,....,x_n\}.$ We divide the proof into two cases, $n=1$ and $n\ge 2$. Notice that, in either case,  the inequality is trivial for any $\beta$ for which $\aln(\beta)=\beta$, so we are free to assume that $\aln(\beta)>\beta$.

If $n=1,$ notice that the Euler-Lagrange equation requires that 
\[
(W_{\alpha,\beta}*\nu)(x_0)\le (W_{\alpha,\beta}*\nu)(0).
\]
More explicitly, as $\nu=\frac{\delta_{x_0}+\delta_{x_1}}{2},$ this inequality reads
$
\frac{1}{2}\left[\frac{1}{\alpha}-\frac{1}{\beta}\right]\le \frac{1}{\alpha 2^\alpha}-\frac{1}{\beta 2^\beta},
$
or, 
\begin{equation}\label{f_1rel}
f_1(\alpha)=\frac{2^{-1}-2^{-\alpha}}{\alpha}\le \frac{2^{-1}-2^{-\beta}}{\beta}=f_1(\beta).
\end{equation}
By definition, $\al=\alo(\beta)$ 
saturates this inequality. Our assumption $\alo(\beta)>\beta$ with the unimodality of $f_1$ from Lemma \ref{unimodality} ensure that for any $\gamma\in (\beta,\alo(\beta)),$ 
\[
f_1(\gamma)>f_1(\beta)=f_1(\alo(\beta)).
\]
This implies that the simplex $\nu$ violates the Euler-Lagrange equation for {$\cE_{W_{\ga,\bt}}$,} and hence that $\gamma \le \alod(\beta)$. Of course, since this inequality holds for all $\gamma\in (\bt,\alo(\bt)),$ this proves that $\alo(\bt)\le \alod(\beta)$ for any $\beta\ge 2$.

Our proof proceeds analogously for $n\ge 2,$ with the key difference being that the definition \eqref{f_n} of $f_n$ is derived from the inequality
\[
(W_{\alpha,\beta}*\psi)(x_0)\le (W_{\alpha,\beta}*\psi)(-x_0),
\]
which again is a necessary condition for the Euler-Lagrange equation  to hold for $\nu$. Since the simplex geometry yields $|x_0|^2 = \frac{n}{2n+2}$ and 
$|x_0+x_1|^2=\frac{n-1}{n+1}$ (c.f.~Theorem \ref{main1} and Remark \ref{R:simplex moments}),
 this {inequality} can be re-expressed as:
\[
\frac{n}{n+1}\left(\frac{1}{\alpha}-\frac{1}{\beta}\right)\le \frac{1}{n+1}\left(\frac{\left(\frac{2n}{n+1}\right)^{\alpha/2}}{\alpha}-\frac{\left(\frac{2n}{n+1}\right)^{\beta/2}}{\beta}\right) +\frac{n}{n+1}\left(\frac{\left(\frac{n-1}{n+1}\right)^{\alpha/2}}{\alpha}-\frac{\left(\frac{n-1}{n+1}\right)^{\beta/2}}{\beta}\right),
\]
or  equivalently, 
\[
f_n(\alpha)=\frac{n-\left(\frac{2n}{n+1}\right)^{\al/2}-n\left(\frac{n-1}{n+1}\right)^{\al/2}}{\al}
\le \frac{n-\left(\frac{2n}{n+1}\right)^{\beta/2}-n\left(\frac{n-1}{n+1}\right)^{\beta/2}}{\beta}=f_n(\beta).
\]
Since $f_n$ is still unimodal for $n\ge 2$, the remainder of the proof proceeds in an identical manner to the  proof for $n=1$ following \eqref{f_1rel}, hence is omitted. 
\end{proof}
We summarize our findings for $n=2$ and $n=1$ in Figures \ref{fig:2d} and \ref{fig:1d} respectively.

\begin{figure}[H]
    \centering
    \includegraphics[scale=0.25]{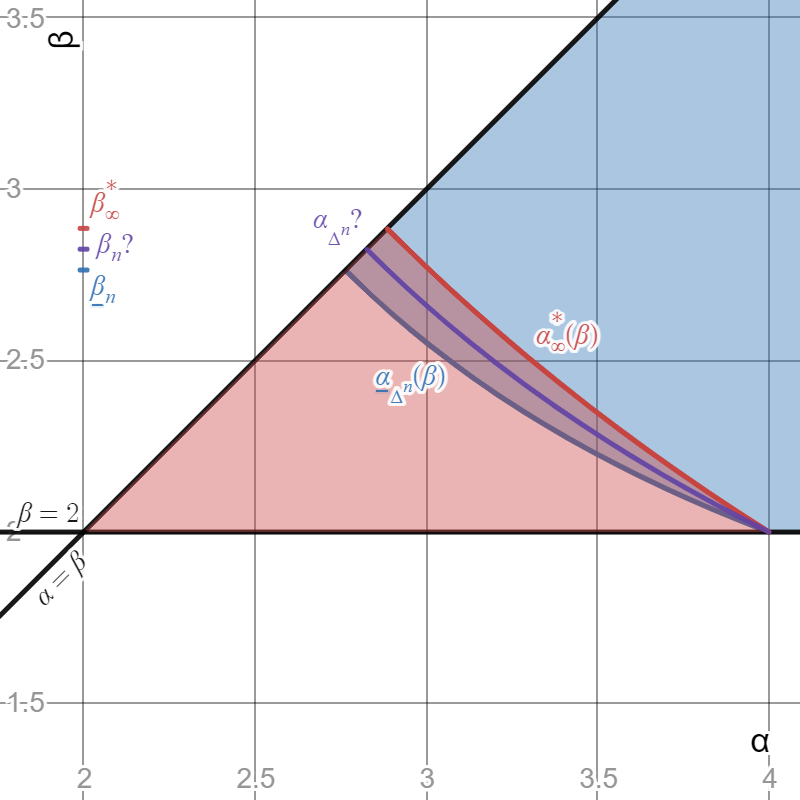}
    \caption{The mildly repulsive regime for, e.g., $n=2$. In the red region to the left of the blue curve $\al=\underline{\alpha}_{\Delta^2}(\beta),$ the simplex does not minimize $\mathcal{E}_{W_{\alpha,\beta}}.$ Conversely, in the rightmost blue region, the simplex uniquely minimizes $\mathcal{E}_{W_{\alpha,\beta}}.$ In the intermediate region, it is not entirely known where the simplex minimizes $\mathcal{E}_{W_{\alpha,\beta}}$, but the graph of the 
threshold function $\alpha_{\Delta^2}$ must lie entirely in this region.}
    \label{fig:2d}
\end{figure}

\begin{figure}[H]
    \centering
    \includegraphics[scale=0.25]{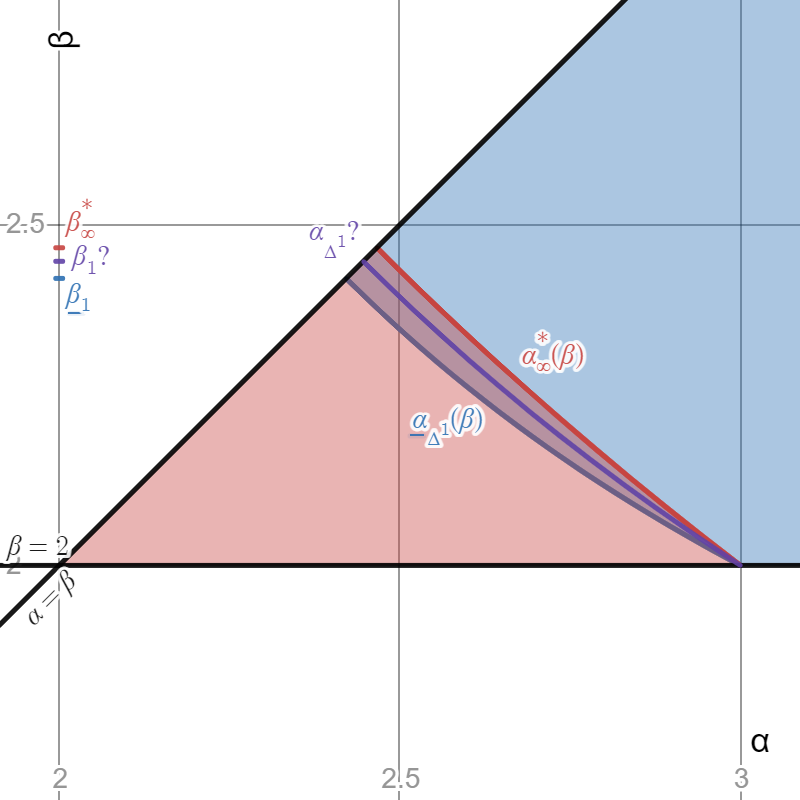}
    \caption{The analogous graph for the $n=1$ case. All coloured regions and graphs have the same meaning as their counterparts in Figure \ref{fig:2d}, although 
the scale of this graph differs from its higher dimensional counterparts, due the fact that $4^*=3$ when $n=1.$}
    \label{fig:1d}
\end{figure}

\noin Notably, even this weaker lower bound tends to the upper bound $\ali$ as $n\to\infty$. }

\begin{proposition}[Bounds converge in the high dimensional limit]\label{P:high dimensional collapse}
For all $\bt \ge 2$, we have\,
$
\lim_{n\to\infty} \aln(\beta) = \ali(\beta) \qquad (n\ne 1).
$
\end{proposition}

\begin{proof}
For $\bt \ge 2$, observe that the unimodal functions $f_n(\bt)$ of Lemma \ref{unimodality} converge to the unimodal limit $\fin(\bt)$ of Remark \ref{R:number of solutions}:
\begin{align*}
\lim_{n \to \infty} f_n(\bt) 
=\lim_{n \to \infty}\frac{n-\left(\frac{2n}{n+1}\right)^{\bt/2}-n\left(\frac{n-1}{n+1}\right)^{\bt/2}}{\bt}
=  1- \frac{2^{\bt/2}}{\bt}
= \fin(\bt) \qquad (n \ne 1).
\end{align*}
Since $\aln(\beta)$ and $\ali(\bt)$ are defined as the largest $\al$ satisfying $f_n(\al)=f_n(\bt)$ and $\fin(\al)=\fin(\bt)$ respectively, 
 it follows that $\aln(\beta) \to \ali(\bt)$ as $n\to \infty$.
\end{proof}

\begin{remark}[Monotonicity]
Numerical experiments displayed in Figure \ref{f:5} suggest $(4-t)(t-2)\aln(t)$ is a non-decreasing function of $n \ge 2$ on $t >0$; for $t \ge 2$ its large $n$ limit is established in
the previous proposition.
To confirm the observed monotonicity rigorously, it would suffice to show that unimodality of $f_{n+1}-f_n$ on $(0,\infty)$ for all $n \ge 2$. This is because, for $n\ge 2,$ $f_n(t)$ has zeroes only at $t=2$ and $t=4$, and hence, assuming unimodality, these are the only two zeroes of $f_{n+1}-f_n$. Since $\lim_{t\to\infty}\left(f_{n+1}(t)-f_n(t)\right)=-\infty$, this implies positivity of $(4-t)(t-2)(f_{n+1}(t)-f_n(t))$ away from $t \in \{2,4\}$.
\end{remark}

\begin{figure}[H]
    \centering
    \includegraphics[scale=0.12]{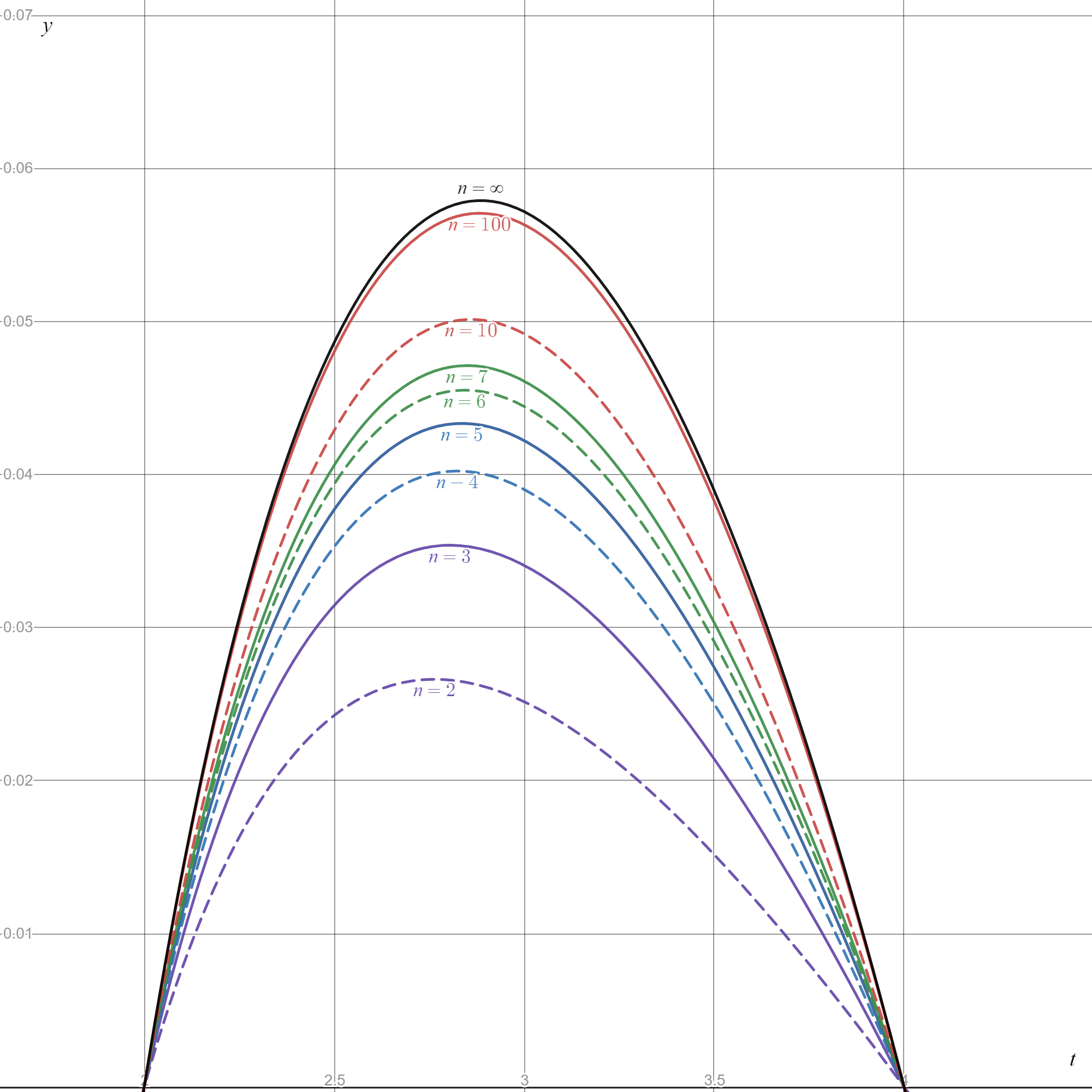}
    \caption{Graphs of $f_n(t)$ for selected values of $n.$ Our numerical experiments indicate that, for all $t\in[2,4],$ $f_n(t)$ increases monotonically to $\fin(t):=1-\frac{2^{t/2}}{t}$.}
    \label{f:5}
\end{figure}

{\em Authors' statement:}   This study does not involve any data.
The authors declare they do not have any conflict-of-interest concerning this manuscript or its contents.

\end{document}